\documentclass[acmsmall]{acmart}
\acmJournal{PACMPL}
\acmVolume{1}
\acmNumber{CONF} % CONF = POPL or ICFP or OOPSLA
\acmArticle{1}
\acmYear{2018}
\acmMonth{1}
\acmDOI{} % \acmDOI{10.1145/nnnnnnn.nnnnnnn}
\startPage{1}

\setcopyright{none}
\usepackage{booktabs}   %% For formal tables:
\usepackage{subcaption} %% For complex figures with subfigures/subcaptions
\usepackage{math-setup}
\usepackage{enumitem}

\begin{document}

\newcommand{\jc}[1]{{\color{red}\{JC: #1\}}}
\newcommand{\jcedit}[1]{{\color{red}{#1}}}
\newcommand{\td}[1]{{\color{red}\{TODO: #1\}}}
\newcommand{\greg}[1]{{\color{red}\{GREG: #1\}}}
\newcommand{\cagatay}[1]{{\color{orange}{[[\c{C}a\u{g}atay]] #1}}}

%% Title information
\title{Data Extraction via Semantic Regular Expression Synthesis}         %% [Short Title] is optional;
                                        %% when present, will be used in
                                        %% header instead of Full Title.
% \titlenote{with title note}             %% \titlenote is optional;
%                                         %% can be repeated if necessary;
%                                         %% contents suppressed with 'anonymous'
% \subtitle{Subtitle}                     %% \subtitle is optional
% \subtitlenote{with subtitle note}       %% \subtitlenote is optional;
%                                         %% can be repeated if necessary;
%                                         %% contents suppressed with 'anonymous'

%% Author information
%% Contents and number of authors suppressed with 'anonymous'.
%% Each author should be introduced by \author, followed by
%% \authornote (optional), \orcid (optional), \affiliation, and
%% \email.
%% An author may have multiple affiliations and/or emails; repeat the
%% appropriate command.
%% Many elements are not rendered, but should be provided for metadata
%% extraction tools.

%% Author with single affiliation.
\author{Qiaochu Chen}
\authornote{Work started and partially completed at Sigma Computing.}
\orcid{0000-0003-4680-5157}
\affiliation{
  \institution{University of Texas at Austin} 
  \city{Austin}
  \state{Texas}
  \country{USA}
}
\email{qchen@cs.utexas.edu}

\author{Arko Banerjee}
\orcid{0000-0001-7688-6133}
\affiliation{
  \institution{University of Texas at Austin}
  \city{Austin}
  \state{Texas}
  \country{USA} 
}
\email{arko.banerjee@utexas.edu}

\author{Çağatay Demiralp}
\authornotemark[1]
\orcid{0000-0002-5933-6620}
\affiliation{
  \institution{MIT CSAIL}
  \city{Cambridge}
  \state{Massachusetts}
  \country{USA} 
}
\email{cagatay@csail.mit.edu}

\author{Greg Durrett}
\orcid{0000-0002-7061-7298}
\affiliation{
  \institution{University of Texas at Austin} 
  \city{Austin}
  \state{Texas}
  \country{USA}
}
\email{gdurrett@cs.utexas.edu}

\author{Işil Dillig}
\orcid{0000-0001-8006-1230}
\affiliation{
  \institution{University of Texas at Austin}
  \city{Austin}
  \state{Texas}
  \country{USA}
}
\email{isil@cs.utexas.edu}

%% Abstract
%% Note: \begin{abstract}...\end{abstract} environment must come
%% before \maketitle command
\begin{abstract}
Many data extraction tasks of practical relevance require not only syntactic pattern matching but also semantic reasoning about the content of the underlying text. While regular expressions are very well suited for tasks that require only syntactic pattern matching, they fall short for data extraction tasks that involve both a syntactic and semantic component. To address this issue, we introduce \emph{semantic regexes}, a generalization of regular expressions that facilitates combined syntactic and semantic reasoning about textual data. We also propose a novel learning algorithm that can synthesize semantic regexes from a small number of positive and negative examples. Our proposed learning algorithm uses a combination of neural sketch generation and compositional type-directed synthesis for fast and effective generalization from a small number of examples.   We have implemented these ideas in a new tool called \toolname and evaluated it on  representative data extraction tasks involving several textual datasets. Our evaluation shows that semantic regexes can better support complex  data extraction tasks than standard regular expressions and that our learning algorithm significantly outperforms existing tools, including  state-of-the-art neural networks and  program synthesis tools.

\end{abstract}

%% 2012 ACM Computing Classification System (CSS) concepts
%% Generate at 'http://dl.acm.org/ccs/ccs.cfm'.
% <concept_id>10011007.10011006.10011008</concept_id>
% <concept_desc>Software and its engineering~General programming languages</concept_desc>
% <concept_id>10003456.10003457.10003521.10003525</concept_id>
% <concept_desc>Social and professional topics~History of programming languages</concept_desc>
% <concept_significance>300</concept_significance>

% \ccsdesc[500]{Software and its engineering~General programming languages}
% \ccsdesc[300]{Social and professional topics~History of programming languages}
%% End of generated code

%% Keywords
%% comma separated list
% \keywords{keyword1, keyword2, keyword3}  %% \keywords are mandatory in final camera-ready submission

%% \maketitle
%% Note: \maketitle command must come after title commands, author
%% commands, abstract environment, Computing Classification System
%% environment and commands, and keywords command.
\maketitle

\section{Introduction}\label{sec:intro}

Regular expressions (or \emph{regexes}) are a convenient and versatile mechanism for extracting information from textual data. Because of their  wide applicability, many programming languages provide  built-in support for regular expressions, allowing developers to perform textual pattern matching. Further, because regular expressions  have numerous applications in user-facing applications like spreadsheets, recent years have seen an explosion in the number of new techniques for learning regular expressions from examples and/or natural language~\cite{alpharegex,regel}.

Despite their general practicality, regexes are mainly applicable in settings where the desired data extraction task is  \emph{purely syntactic} in nature. For example,  regexes are  very well-suited to tasks like describing phone numbers and dates  because such concepts can be described in terms of a specific syntactic format (e.g., \textsf{+D-DDD-DDD-DDDD} or \textsf{DD/DD/DDDD}). However, many data extraction tasks of practical relevance are \emph{not} so easy to describe using a \emph{purely} syntactic pattern. As a simple example, consider the task of  extracting \emph{business} emails from a text file. Any email address must follow a certain syntactic format, but this task also involves a \emph{semantic} component in that it requires determining whether some text in the email describes a business entity. As another example, consider the problem of extracting zip codes that fall within a certain range. In addition to checking whether a string syntactically matches a zip code pattern (DDDDD or DDDDD-DDDD), it requires interpreting part of the string as a number and then performing a semantic range check, which is difficult to do using regexes.

Based on this observation, this paper proposes the concept of a \emph{semantic regex} as a mechanism for combining the strengths of syntactic pattern matching with semantic reasoning. Our proposed semantic regexes generalize standard regular expressions in that they provide a \emph{semantic pattern matching construct} which accepts strings that (a) belong to a category $\tau$ (e.g., business, location, person) and (b) satisfy a predicate $\phi$ when interpreted as an instance of type $\tau$. For example, this construct can be used to match strings that (a) correspond to a \textsf{City} (type $\tau$), and (b) further satisfy some additional criterion, such as being in the United States or in the state of California (predicate $\phi$). Under the hood,  semantic pattern matching employs large language models like GPT-3~\cite{gpt3, palm} to test membership in some category $\tau$ but further allows refining the query result using a logical predicate $\phi$. In this sense, one can view our semantic regexes as deciding  membership in a refinement type and then combining the matching strings using  standard regex operators.

Beyond proposing the notion of semantic regexes, another key contribution of this paper is a new synthesis algorithm for learning semantic regexes from positive and negative examples. The learning problem in this context is more challenging than traditional regex synthesis because semantic regexes are much more expressive than standard regexes. As a result, the hypothesis space in this setting is very large, which has two important consequences:

\begin{itemize}[leftmargin=*]
\item First, the semantic regex learning problem cannot be solved using a purely search-based approach due to the sheer size of the search space. In fact, the search space is theoretically not even bounded because our semantic regex language does not restrict the types $\tau$ to a pre-defined vocabulary.
\item  Second, due to the extremely large hypothesis space, there are typically many semantic regexes consistent with a small number of examples. Hence, to find the \emph{intended} semantic regex, our learning algorithm must have a strong inductive bias towards user intent. 
\end{itemize}

The synthesis technique proposed in this paper surmounts these challenges using a novel combination of three key ideas:

\begin{enumerate}[leftmargin=*]
    \item {\bf Neural sketch generation:} Our learning algorithm uses a large language model (GPT-3) to generate a sketch of the desired semantic regex. Our key observation is that LLMs are well-suited to this task because they are effective at identifying semantic commonalities between the positive examples and {inferring appropriate types to be used within the semantic pattern matching constructs.}
    % \cagatay{Consider revising the following for clarity}conjecturing the types that should be used inside the semantic pattern matching constructs. 
\item {\bf Compositional synthesis:} Our learning algorithm decomposes the synthesis task into multiple simpler sub-problems. Because the holes (i.e., unknowns) in the generated sketches are \emph{typed}, the synthesis technique lends itself to a compositional solution, where we can synthesize each hole largely (though not entirely) independently. 
\item {\bf Type-directed search:} The presence of type information in the sketches makes it possible to fill each hole in a type-directed way. Specifically, we utilize a type system with subtype polymorphism to infer the space of \emph{valid} completions of a hole.
\end{enumerate}

 \begin{figure}
 \vspace{-0.5cm}
     \includegraphics[scale=0.22, trim=40 80 50 650, clip]{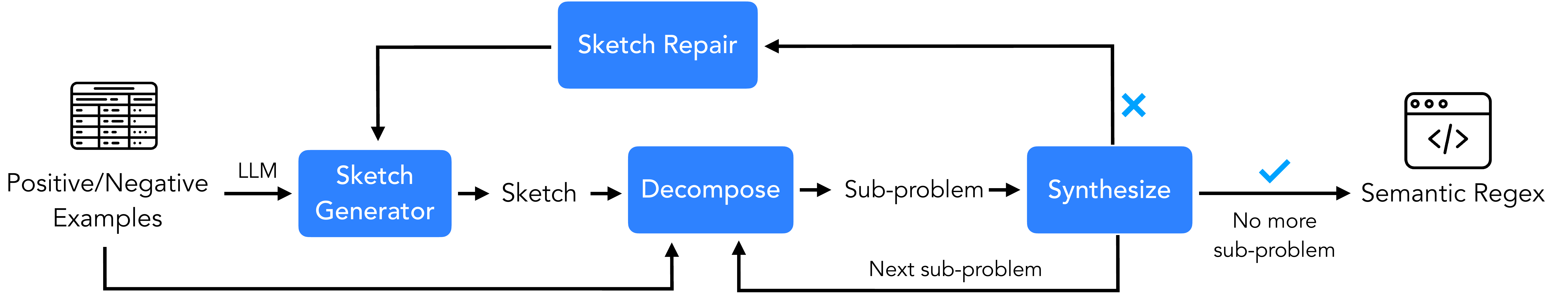}
     \vspace{-0.1in}
     \caption{Schematic overview of our approach.}
     \label{fig:workflow}
     \vspace{-0.5cm}
 \end{figure}

 Figure~\ref{fig:workflow} shows the workflow of our proposed learning approach, which first utilizes the provided examples to generate a semantic regex \emph{sketch} using GPT-3. In the next step, our approach searches for completions of the sketch by (a) decomposing the overall problem into several subproblems and (b) using type-directed synthesis to solve each subproblem. If the sketch has a valid completion, the resulting semantic regex is returned to the user. Otherwise, our approach analyzes the root cause of failure and uses this information to query the language model for a more accurate sketch. 

 We have implemented the proposed technique in a tool called \toolname  and evaluated it on  information extraction tasks involving several different datasets. Our evaluation shows that these data extraction tasks can be successfully automated using our proposed semantic regexes and that our learning algorithm is quite effective for automating the desired data extraction task. In particular, our approach achieves an average $F_1$ score of $0.87$ on the test data, while prior data extraction techniques achieve a maximum $F_1$ score of $0.65$.

To summarize, this paper makes the following contributions:

\begin{itemize}[leftmargin=*]
    \item We propose \emph{semantic regular expressions} to combine the flexibility of syntactic pattern matching with semantic queries involving types and logical predicates.
    \item We describe a new learning technique for synthesizing semantic regexes from positive and negative examples. Our approach combines the power of large language models with type-directed synthesis for effective automation of data extraction tasks.
    \item We evaluate our tool, \toolname, on representative data extraction tasks and show that semantic regexes are useful for these tasks and that our learning approach outperforms other data extraction techniques in terms of  average $F_1$ score. 
\end{itemize}

\section{Overview}\label{sec:overview}

\begin{figure}[!t]
    \centering
    \footnotesize
    \begin{tabular}{ccc}
        \toprule
        Object ID & Department & Artist Bio \\
        \midrule
        1 & Decorative Art & Heinrich Reinhold, Germany, 1740-1789  \\
        2 & Contemporary Art & Cindy Sherman, United States, 1954-present \\
        3 & Medieval Art & Sandro Botticelli, Italy, 1470-1561 \\
        4 & Modern Art & Max Ernst, Germany, 1891-1976 \\
        5 & Medieval Art & Niclaus Gerhaert von Leyden, North Netherlands, 1462-1473 \\
        $\cdots$ & $\cdots$ & $\cdots$ \\
        \bottomrule
    \end{tabular}
    \vspace{-0.2cm}
    \caption{Dataset about pieces of art exhibited in a museum.}
    \label{fig:dataset_example}
    \vspace{-0.5cm}
\end{figure}

In this section, we illustrate our technique using the motivating example  shown in Figure~\ref{fig:dataset_example}, which contains information about artworks exhibited at a museum. Given this dataset, suppose that a user wants to extract all European artists who were born before the 20th century and whose  name contains Thomas. This data extraction task is challenging because it requires both syntactic and semantic reasoning:

\begin{itemize}[leftmargin=*]
    \item {\bf Syntax:} In order to retrieve the desired information from this dataset, we first need to perform pattern matching over the syntax of the ``Artist Bio'' column. In particular, because this column contains information of the form ``Name, Country, Birth Year - Death Year'', we first need to \emph{syntactically} parse the input string into its four constituent fields and check whether the first field (corresponding to the artist name) contains ``Thomas''.
    \item {\bf Semantics:} After performing syntactic pattern matching, we then need to perform semantic reasoning about the contents of each row to understand whether (a) the first field describes a name, (b) the artist's nationality is European and (c) they were born before the 20th century.
\end{itemize}

\subsection{Semantic Regexes} 
Our proposed \emph{semantic regex} concept is a natural fit for the data extraction task illustrated in this example. Semantic regexes combine the convenience of regexes for syntactic pattern matching with the power of semantic reasoning about  data types. In addition to supporting the standard regex operators (concatenation, disjunction, Kleene star), semantic regexes provide the following  \emph{semantic} pattern matching construct, written using a refinement-type-like notation:
\[
\footnotesize
\{ v: \tau \ | \ \phi \}
\]
This construct matches any string that  is semantically of type $\tau$ and that further satisfies the (optional) logical qualifier $\phi$. For instance, going back to our example, recall that we need to pattern match strings that  correspond to a European country. This can be  expressed using the semantic regex $
\{ v: \textsf{Country} \ | \ v \in \textsf{Europe} \}
$, which, for example, matches the strings ``France'', ``Britain'' and ``North Netherlands'', but fails on the strings ``United States'', ``Korea'' etc. Similarly, we can express the desired constraint on the artists' birth year using the following semantic regex:
\[
\footnotesize
\{ v: \textsf{Year} \ | \ v < 1900 \}
\]
which matches strings that (a) correspond to a year and (b) whose value is less than or equal to 1899. Putting all of this together, our desired data extraction task can be accomplished using the following overall semantic regex:
% \begin{figure}[H]

\vspace{-0.5cm}
\footnotesize
\begin{align*}
    & \regex_1 \cdot  ``, \ " \cdot \regex_2  \cdot ``, \ " \cdot \regex_3 \cdot ``-" \cdot \regex_4 \\
    \texttt{where } & \regex_1 = \{ v: \textsf{Name} \} \ \cap \ {\tt Contain}(``{\tt Thomas}") \\
                    & \regex_2 = \{ v: \textsf{Country} \ | \ v \in \textsf{Europe}\} \\
                    & \regex_3 = \{ v: \textsf{Year} \ | \ v< 1900\} \\
                    & \regex_4 = \{ v: \textsf{Year} \}
\end{align*}
\normalsize
\vspace{-0.5cm}

% \caption{Ground truth program for the task}\label{prog:gt}
% \vspace{-0.5cm}
In other words, this semantic regex matches all strings of the form ``X, Y, Z-W" where $X$ is a name containing Thomas, $Y$ is a European country, $Z$ is a year before 1900, and $W$ is any year. 

\subsection{Synthesizing Semantic Regexes} While semantic regexes provide a useful mechanism for information extraction, they can nonetheless  be non-trivial for end-users to construct. Motivated by this problem, another key contribution of this paper is a new technique for synthesizing semantic regexes from a small number of positive and negative examples. We now illustrate how our technique can be used to automate the data extraction task for our running example. 
Suppose that the user describes the target data extraction task using the following positive and negative examples:

\begin{center}
\footnotesize
\begin{tabular}{cc}
    \toprule
     {\bf Positive Examples} & {\bf Negative Examples}  \\
     \midrule
     John Thomas Young Gilroy, Britain, 1898-1985 &  Alma Thomas, United States, 1891-1978 \\ 
     Thomas Hudson, Britain, 1701-1779 & Sandro Botticelli, Italy, 1470-1561 \\
     Thomas Couture, France, 1815-1879 & Thomas Nölle, Germany, 1948-2020 \\
     \bottomrule
\end{tabular} 
\end{center}

Here, the positive examples correspond to the artist biographies that should be extracted, while the negative examples are those that should be ignored. In particular, the first negative example does not conform to the ``European country'' restriction; the second negative example does not contain ``Thomas'' in the artist's name;  and the third one fails the criteria ``born before the 20th century''.  We will now describe how our approach synthesizes the target semantic regex given only these examples.

 At the heart of our learning approach lies the notion of a \emph{typed sketch}, which captures the general {syntactic} structure of the target semantic regex. In addition, the holes (i.e., unknowns) in the sketch are annotated with types capturing commonalities in the positive examples.  
Returning to our running example, our synthesis approach generates an initial candidate sketch by querying a large language model (GPT-3) with user-provided positive examples. Suppose that GPT-3 returns the following sketch:
\[
\footnotesize
\thole{Name}\cdot ``, \ " \cdot\thole{Country}\cdot ``, \ " \cdot\thole{Year}
\]
Here, the symbol $\thole{}$ denotes an unknown expression, and the notation $\thole{t}$ indicates that any string matched by $\thole{}$ should be a subtype of ${\sf t}$.
% \gd{seems like these types are different than the Int/Artist types above. the font illustrates a difference but this wasn't that clear}

Starting with the GPT-3-synthesized sketch, our method decomposes the synthesis problem into multiple sub-problems, one for each hole in the sketch, and  performs a type-directed search to complete each hole. For this example,  our synthesis method infers the following positive examples for each hole:

\begin{center}
    \footnotesize
    \begin{tabular}{cccc}
        \toprule
        $\thole{Name}$ & $\thole{Country}$ & $\thole{Year}$  \\
         \midrule
         John Thomas Young Gilroy & Britain & 1898-1985 \\ 
         Thomas Hudson & Britain & 1701-1779 \\
         Thomas Couture & France & 1815-1879 \\
         \bottomrule
    \end{tabular} 
\end{center}

Note that it is not possible to propagate  negative examples for individual holes, as it suffices for the synthesized regex for \emph{one}  hole to reject its corresponding string, but we do not a priori know which one.  In particular, for this  example, it would \emph{not} be accurate to deduce that ``Alma Thomas'', ``Sandro Botticelli'', and ``Thomas Nölle'' as  negative examples for the first hole.

Given this decomposition, our approach tries to synthesize a regex $\regex_i$ for each hole $\{\hole: \type_i\}_i$ such that (a) the type of $\regex_i$ is a subtype of $\type_i$ and (b) $\regex_i$ matches all of its corresponding positive examples. For this example, our synthesis algorithm can immediately deduce that the sketch is incorrect since no subtype of $\mathsf{Year}$ can match the corresponding positive examples for the third hole.

 To repair the sketch, our learning algorithm localizes parts of the sketch for which synthesis failed (in this case, \textsf{Year}) and synthesizes a different sketch for the failing part. In the next iteration, suppose that we consider the following correct sketch:
\[
\footnotesize
\thole{Name}\cdot ``, \ " \cdot\thole{Country}\cdot ``, \ " \cdot\thole{Year}\cdot ``-" \cdot\thole{Year}
\]

Our synthesis algorithm tries to independently find the completion of each hole with the appropriate type and satisfy the corresponding decomposed positive examples. As before, the positive examples are used to prune the search space: for example, since the second hole must match the strings ``Britain'' and ``France'', the synthesizer can rule out completions such as
$
\{v: \textsf{Country} \ | \ v \in \textsf{Asia} \} 
$
and
$
\{v: \textsf{Country} \ | \ v \in \textsf{Asia} \ \land \ldots \} 
$. Similarly, type information in the sketch is critical, enabling the synthesizer to avoid enumerating useless sub-programs. For instance, when synthesizing the last hole in the sketch, the synthesizer would not enumerate programs such as $\{ v: \mathsf{Month} \ | \ldots  \} \cup \{v: \mathsf{Date}  \ | \ldots \} $, since this regex can match strings that are not of type $\mathsf{Year}$. It would, however, consider regexes of the form $\{ v: \mathsf{Year} \ | \ v \leq \ldots \}$, as the strings that are matched by this regex would be a subtype of year. After independently synthesizing each hole, the  algorithm checks whether the resulting regex $\regex$ rejects all negative examples and, if so, returns $\regex$ as a solution. Otherwise, it generates a different regex by looking for a different completion for at least one of the holes.

\section{Semantic Regular Expressions }\label{sec:dsl}

In this section, we describe the syntax and semantics of our proposed semantic regular expression language. At a high level, semantic regexes combine standard regular expression operators with pre-trained neural networks that identify semantic types and provide knowledge about the world. 

\begin{figure}[t]
\vspace{-0.5cm}
\small
\[
\begin{array}{rl}
\rho ::= & \lambda s. \ {\tt match}(s, \regex) \\

\regex ::= & c \ | \ cc   \ | \ \emptyset \\
& | \ \matchsemq{(\type_q)}{f}  \ | \ \matchsem{(\type_b)}{f}{\phi}  \\ 
& | \ \neg \regex \ | \ \regex? \ | \ \regex* \ | \ \regex+ \ | \ \regex\{k_1\} \ | \ \regex\{k_1, k_2\} \\
& | \ \regex \cdot \regex \ | \ \regex \cup \regex \ | \ \regex \cap \regex \\

f ::= &  {\tt id} \ | \ {\tt toUpper} \ | \ {\tt toLower} \ | \ {\tt abbreviate}[c] \\ 

\phi ::= &  \top \ | \ \neg \phi \ | \ \phi \wedge \phi \ | \ \phi \vee \phi \\ 
        & | \  t \oplus_{\type_b} t   \ \ {\tt where} \ \oplus \in \{\leq, \geq, =, \in \}\\

t ::=  & v \ | \ v.a \ | \ c \ | \ n \\

\type_b ::= & {\sf Person} \ \mid \ {\sf Organization} \ \mid \ {\sf Product} \ \mid \ {\sf Event} \ \mid \ {\sf Work\ of\ Art}  \\
        & \mid  {\sf Number} \ \mid \ {\sf Integer} \ \mid \ {\sf Float} \\
        & \mid \ {\sf Date} \ \mid \ {\sf Year} \ \mid \ {\sf Month} \ \mid \ {\sf Day} \\ 
        & \mid \ {\sf Time} \ \mid \ {\sf Hour} \ \mid \ {\sf Minute} \ \mid \ {\sf Second} \\ 
        & \mid \ {\sf Place} \ \mid \ {\sf Location} \ \mid \ {\sf Nationality} \ \mid \ {\sf Country} \ \mid \ {\sf City} 

\end{array}
\]
\vspace{-0.5cm}
\caption{Semantic string matching language. $c$ is a constant string, $cc$ is a character class (e.g. letters). $\type_b$ is a built-in base type,  and $\type_q$ is an arbitrary base type in our type system.  Also, $k \in \mathds{Z} $, $n \in \mathds{R}$, and $a \in {\tt Attributes}$, where $\tt Attribute$ is type-dependent. }
\label{fig:dsl}
\vspace{-0.5cm}
\end{figure}

\paragraph{\bf DSL Syntax} The syntax of our semantic string matching language  is presented in Figure~\ref{fig:dsl}. A semantic regex $\rho$ takes as input a string $s$ and returns a boolean indicating whether there is a match. Semantic regexes include all the standard regular expression constructs, including constant strings $c$, character classes like letters and numbers (denoted $cc$), concatenation ($\cdot$), complement ($\neg$), union ($\cup$), intersection ($\cap$),  and Kleene star $(*)$. Additionally, the notation $r\{k_1\}$ denotes repetition of $r$ $k_1$ times and $r \{ k_1, k_2\}$ denotes $r$ repeated between $k_1$ to $k_2$ times. As standard, $r?$ indicates an optional occurrence of $r$, and $r+$ denotes one or more occurrences of $r$.

In addition to these standard regex constructs, Figure~\ref{fig:dsl} includes two \emph{semantic pattern matching constructs}, denoted as $\{ v: f(\tau_q) \}$ and $ \{ v: f(\tau_b) \ | \ \phi \}$, where $f$ is an (optional) built-in function, $\tau_b$ is a built-in type (\textsf{Integer}, \textsf{Month}, etc) and $\tau_q$ is an \emph{arbitrary} (user-defined) type. Note that the DSL does not place any restrictions on $\tau_q$, so the user can provide any arbitrary string to define their own type. However, we only allow a logical qualifier $\phi$ to be used for built-in types.

In the most basic form, the construct $\{ v: \tau \}$ matches strings that are semantically of type $\tau$, where $\tau$ can either be a built-in or user-defined type. For example, $\{ v: \textsf{Place} \} $ matches any string that corresponds to a geographical location. The optional function $f$ used in this construct allows refining the query result by performing additional semantic-preserving string processing. For example, $\{ v: {\tt toUpper}(\textsf{Place}) \} $ matches any string that corresponds to a location name in upper case letters (e.g., ``NEW YORK''). More generally, $\{ v: f(\tau) \}$ matches a string $s$ if $s$ is equal to $f(s')$ where $s'$ is a string of type $\tau$. As another example, $\{ v: {\tt abbreviate}[.](\textsf{Place}) \}$ matches the strings ``N.Y.'', ``S.F.'' etc. because the function ${\tt abbreviate}[c]$ abbreviates a string through initialism, using the character $c$ as a separator. 

When performing semantic pattern matching using built-in types $\tau_b$, one can additionally use a logical qualifier $\phi$. In particular, $\{ v: \tau_b \ | \ \phi  \} $ matches those strings that are of type $\tau_b$ and additionally satisfy predicate $\phi$. To check whether a string $s$ satisfies $\phi$, $s$ is first parsed as an instance $o$ of type $\tau_b$ and then checked for conformance against $\phi$. Note that these semantics justify why logical qualifiers are only allowed with built-in types: because we need to parse the string as an instance of $\tau_b$, there must be some built-in mechanism for deserializing the string, which only makes sense for pre-defined types. As an example,  the semantic 
 regex $\{ v: \textsf{Float} \ | \ v < 0.1 \} $ matches strings that can be interpreted as a floating point number whose value is less than $0.1$ (e.g., $0.0051$).  As another example, $\{ v: {\tt toUpper}(\textsf{City}) \ | \ v \in \textsf{Europe} \} $ matches strings, such as ``ROME'' that (a) correspond to European cities and (b) are in upper case letters.

\begin{figure}[t]
\vspace{-0.5cm}
\small
\[
\begin{array}{r l}
\sem{\lambda s. \ {\tt match}(s, \regex)}s = & s \in \sem{\regex} \\
\sem{c} = & \{c\} \\
\sem{\neg \regex} = & \{s \ | \ s \notin \sem{\regex}\}\\
\sem{\regex^0} = & \{\epsilon\} \\ 
\sem{\regex^i} = & \{s_1 \cdot s_2 \ | \ s_1 \in \sem{\regex^{i-1}}, s_2 \in \sem{\regex}\} \\ 
\sem{\regex *} = & \bigcup_{n \in \{0..\infty\}} \sem{\regex^n}  \\
\sem{\regex_1 \cdot \regex_2} = & \{s_1 \cdot s_2 \ | \ s_1 \in \sem{\regex_1}, s_2 \in \sem{\regex_2} \} \\
\sem{\regex_1 \cup \regex_2} = & \sem{\regex_1} \cup \sem{\regex_2} \\ 
\sem{\regex_1 \cap \regex_2} = & \sem{\regex_1} \cap \sem{\regex_2} \\ 
\sem{\matchsemq{(\type_q)}{f}} = & \{\sem{f}s \ | \ {\tt SemanticType}(s) = \type_q\}  \\ 
\sem{\matchsem{(\type_b)}{f}{\phi}} = & \{\sem{f}s \ | \ {\tt SemanticType}(s) = \type_b \wedge {\tt Cast}{\tt <}\type_b{\tt >}(s) = o \wedge \phi(o)\}
\end{array}
\]
\vspace{-0.5cm}
\caption{Semantics of matching part of the DSL. Here, {\tt SemanticType} is an oracle that determines the semantic type of string $s$, {\tt Cast<}$\type${\tt >} casts string $s$ to object $o$ of type $\type$.}  \label{fig:dsl_semantics_match} 
\vspace{-0.5cm}
\end{figure}

\paragraph{\bf DSL Semantics} Figure~\ref{fig:dsl_semantics_match} presents the formal semantics of  our DSL for semantic string matching, where  $\sem{r}$ denotes the set of all strings that $r$ matches.\footnote{Semantics of functions are provided in the appendix.} 
%In the following discussion, we only focus on the semantic pattern matching constructs, as the semantics of the other regex operators are completely standard. 
Observe that the semantics of the DSL is parametrized by a helper function called {\tt SemanticType}, which is implemented by a pre-trained neural network and which is used to check whether the type of a string $s$ is $\tau$. Hence, the construct $\{ v: f(\tau_q)\}$ matches all strings $s$ such that (a) $s = f(s')$ for some string $s'$, and (b) where ${\tt SemanticType}(s') = \tau_q$.  Similarly, $\{ v: f(\tau_b) \ | \ \phi \}$ matches all strings $s$ such that (a) $s = f(s')$ for some string $s'$, (b) $s'$ is an instance of built-in type $\tau_b$, and (c) when $s'$ is parsed into an object $o$ of type $\tau_b$, o satisfies predicate $\phi$.

\begin{example}
 The semantic regex $\{ v: \textsf{Date} \ | \ v.\textsf{month} = 5\}$ matches all strings that represent dates in May. In particular, any string matching a Date is first parsed into a datetime object and its \textsf{month} field is checked for being equal to 5.  Examples of strings matched by this regex include ``May 2023'' and ``2023-05-01''. 
\end{example}
\section{Overview of the Type System}\label{sec:type}

While our semantic regex DSL is not \emph{explicitly} typed, our approach utilizes a type system to facilitate effective synthesis. In this section, we give an overview of the type system. 

\subsection{Type Syntax}\label{subsec:type_syntax}

\begin{figure}[t]
% \vspace{-0.5cm}
\small
\[
\begin{array}{r l}
\type := & {\sf Any}   \ \mid \ \opttype{\type'} \ \mid \ \type' \\
\type' := & \semtype{\type_s} \ \mid \ {\sf CharSeq} \\ 
\type_s := & {\sf Person} \ \mid \ {\sf Organization} \ \mid \ {\sf Product} \ \mid \ {\sf Event} \ \mid \ {\sf Work\ of\ Art}  \\ 
        & \mid \ {\sf Number} \ \mid \ {\sf Integer} \ \mid \ {\sf Float} \\
        & \mid \ {\sf Date} \ \mid \ {\sf Year} \ \mid \ {\sf Month} \ \mid \ {\sf Day} \\ 
        & \mid \ {\sf Time} \ \mid \ {\sf Hour} \ \mid \ {\sf Minute} \ \mid \ {\sf Second} \\ 
        & \mid \ {\sf Place} \ \mid \ {\sf Location} \ \mid \ {\sf Nationality} \ \mid \ {\sf Country} \ \mid \ {\sf City} 
\end{array}
\]
\vspace{-0.5cm}
\caption{Type syntax. }  
\label{fig:type_syntax}
\vspace{-0.5cm}
\end{figure}

The syntax of our type system is shown in Figure~\ref{fig:type_syntax}, where \textsf{Any} corresponds to the top element in the type system and ${\sf CharSeq}$ indicates any string without semantic meaning, such as ``1a2b3c'', ``,.3d,.'' etc. The type $\semtype{\type_s}$ indicates strings that can interpreted as instance of $\type_s$ (e.g., \textsf{Date}). In addition, the type $\opttype{\type}$ includes both $\epsilon$ (empty string) as well as any string of type $\type$. Semantic types $\tau_s$ include both built-in types $\tau_b$ (e.g., \textsf{Integer}, \textsf{Float}, \textsf{Date}) as well as user-defined types $\tau_q$. Hence, the type syntax is \emph{not} fixed a priori and is parametrized over any user-defined types that occur in the program.

\subsection{Subtyping}

Our type system supports subtype polymorphism because there is a natural subtyping relation between many entities of interest.  We formalize the subtyping relation in Figure~\ref{fig:subtyping} using the standard judgment $\vdash \type_1 \subtype \type_2$, indicating that $\type_1$ is a subtype of $\type_2$. In Figure~\ref{fig:subtyping}, the first three rules are straightforward and establish \textsf{Any} as the top element of the type system. The following rules (until {\sc Trans}) show the subtyping relation involving built-in semantic types. For example, according to these rules, \textsf{Year}, \textsf{Month}, and \textsf{Day} are all subtypes of the more generic \textsf{Date} type. The {\sc Trans} rule states the transitivity of the subtyping relation and the {\sc Semantic} rule lifts the subtyping relation to ${\tt Semantic}(\tau)$. The last two rules for {\sc Optional} are also standard: {\sc Optional-Width} states that any type $\tau$ is a subtype of ${\tt Optional}(\tau)$ and the last rule lifts the subtyping relation to optional types. Finally, the last rule handles subtyping between user-defined types. If the set of objects represented by  $\tau_1$ is a  subset of those represented by $\tau_2$, we have $\tau_1 \subtype \tau_2$. In practice, we perform this check by querying a semantic ontology (specifically, DBPedia~\cite{dbpedia} in our implementation). 

\begin{figure}[t]
\vspace{-0.5cm}
\[
\small
\begin{array}{lllll}
    & \vdash {\sf CharSeq} \subtype {\sf Any} 
    & \vdash \semtype{\type_s} \subtype {\sf Any} 
    & \vdash \opttype{\type} \subtype {\sf Any}  \\
    %& \textbf{Semantic Types Subtyping Relations} \\\\
    & \vdash {\sf Year} \subtype {\sf Date}  
    & \vdash {\sf Month} \subtype {\sf Date} 
    & \vdash {\sf Day} \subtype {\sf Date} \\
    & \vdash {\sf Hour} \subtype {\sf Time} 
    & \vdash {\sf Minute} \subtype {\sf Time} 
    & \vdash {\sf Second} \subtype {\sf Time} \\
    & \vdash {\sf Country} \subtype {\sf Place} 
    & \vdash {\sf City} \subtype {\sf Place} \\ 
    & \vdash {\sf Institution} \subtype {\sf Organization}  
    & \vdash {\sf Company} \subtype {\sf Organization} \\
\end{array}    
\]
\footnotesize
\begin{mathpar}
    \inferrule*[Left=Trans]{\vdash \type'' \subtype \type' \ \ \ \vdash \type' \subtype \type}{\vdash \type'' \subtype \type} \and
    \inferrule*[Left=Semantic]{\vdash \type' \subtype \type}{\vdash \semtype{\type'}  \subtype \semtype{\type}} \\
    \inferrule*[Left=Optional-Width]{\ \ \ \ }{\vdash \type \subtype \opttype{\type}} \\
    \inferrule*[Left=Optional-Congruence]{\vdash \type' \subtype \type}{\vdash \opttype{\type'} \subtype \opttype{\type}} \and
    \inferrule*[Left=User-defined]{ \gamma(\type') \subseteq \gamma(\type)}{\vdash \type' \subtype \type}
\end{mathpar}
\vspace{-0.5cm}
\caption{Subtyping relations. $\gamma(\tau)$ is the concretization function denoting the set of objects represented by $\tau$.}
\label{fig:subtyping}
\vspace{-0.5cm}
\end{figure}

\subsection{Typing Rules}

We present the typing rules for assigning types to DSL terms in Figure~\ref{fig:typing_rules}. These rules derive judgments of the form $\vdash t: \type$ indicating that term $t$ has type $\type$. Note that Figure~\ref{fig:typing_rules} only shows a representative subset of the typing judgments;  the full set is presented in the Appendix under supplementary materials. 
% \cagatay{Consider removing ``Typing rules for'' from the subhead titles below, e.g., Constant and characters,..., Negation and concatenation.}
\begin{figure}[ht]
% \vspace{-0.5cm}
\footnotesize
\begin{mathpar}
    \inferrule*[Left=Const-Semantic]{ {\tt SemanticType}(c) = \type\\\\\type \neq {\sf CharSeq}}{\vdash c: {\tt Semantic}(\type)}\and \ \ \ \
    \inferrule*[Left=Const-CharSeq]{{\tt SemanticType}(c) = {\sf CharSeq} }{\vdash c: {\sf CharSeq}}\\
    \inferrule*[Left=CC]{ cc \neq {\tt <Num>} }{\vdash cc: {\sf CharSeq}} \and
    \inferrule*[Left=CC-Num]{ cc = {\tt <Num>} }{\vdash cc: {\tt Semantic}({\sf Number})} \\ 
    \inferrule*[Left=matchSem]{ \ \ \ }{\vdash \matchsem{(\type_b)}{f}{\phi} : {\tt Semantic}(\type_b)} \and
    \inferrule*[Left=matchSem]{ \ \ \ }{\vdash \matchsemq{(\type_q)}{f} : {\tt Semantic}(\type_q)} \\
    \inferrule*[Left=Lifting]{\vdash r: \type \ \ \  \type \subtype \type'}{\vdash r: \type'} \and
    \inferrule*[Left=Optional]{\vdash r: \type}{\vdash r?: {\tt Optional}(\type)} \\
    \inferrule*[Left=Union]{\vdash r_1: \type_1 \ \ \ \vdash r_2: \type_2}{\vdash r_1 \cup r_2 : \type_1 \join \type_2} \and
    \inferrule*[Left=And]{\vdash r_1: \type_1  \ \ \ \vdash r_2: \type_2}{\vdash r_1 \cap r_2 : \type_1 \meet \type_2} \\
    \inferrule*[Left=Not]{\vdash r: \type}{\vdash \neg r: {\sf Any}} \and
    \inferrule*[Left=Concat]{\vdash r_1: \type_1 \ \ \ \vdash r_2: \type_2}{\vdash r_1 \cdot r_2 : {\sf Any}}

\end{mathpar}
\caption{Typing rules.}    
\label{fig:typing_rules}
\vspace{-0.5cm}
\end{figure}

\paragraph{\bf Constant and characters} The first four rules show how to assign types to string constants  and character classes. For constants, we determine their type by querying a semantic oracle (GPT-3 in our implementation) and assign  ${\sf CharSeq}$ if the oracle does not return a semantic type.\footnote{The semantic oracle returns ${\sf CharSeq}$ if the string has no semantic meaning.} Character classes only have semantic meaning for numbers, so we assign the ${\tt Semantic}({\sf Number})$ type if the character is a number, and ${\sf CharSeq}$ otherwise.

\paragraph{\bf Semantic matching} The \textsc{MatchSem} rules present  the typing rules for the semantic matching construct. The type of the expression is identical to the type specified as part of the program syntax.

\paragraph{\bf Union and intersection} The typing rules for union and intersection presented in the \textsc{Union} and \textsc{And} rules, respectively. These rules utilize the   $\join$ and $\meet$ operators, which are defined in Figure~\ref{fig:type_intersection}. At a high level, the meet and join of two types are determined as the least upper bound ($\sqcup$) and the greatest lower bound ($\sqcap$), respectively, in the corresponding type lattice. However, there is a special case for the ${\sf CharSeq}$ type: Intuitively, taking the intersection of a semantic type $\tau$ and \textsf{CharSeq} further refines the objects of type $\tau$ by placing an \emph{additional} syntactic restriction; hence, ${\tt Semantic}(\tau) \land \textsf{CharSeq}$ is defined as ${\tt Semantic}(\tau)$. In contrast, the join of ${\tt Semantic}(\tau)$ and $\textsf{CharSeq}$ is the top element \textsf{Any}, as expected.

\paragraph{\bf Not and concatenation} The \textsc{Not} and \textsc{Concat} are two cases where specific types cannot be inferred. Even though the type of their arguments is known, the resulting type cannot be determined, resulting in an output type of ${\sf Any}$.
%\footnote{In some cases we can actually infer more specific types. For instance, the concatenation of type {\tt First name} and {\tt Last name} has the output type {\tt Name}. However, since rules like this require domain knowledge, we only provide general typing rules in the paper.}.

% \cagatay{Examples provided in prior sections are quite useful. Consider one for this section as well. Perhaps for a few select rules above.}

\begin{figure}[ht]
\vspace{-0.5cm}
\small
    \[
    \begin{array}{rl}
        \type \meet {\sf Any} = & \type \\
        \type_1 \meet \opttype{\type_2} = & \type_1 \meet \type_2 \\ 
        \opttype{\type_1} \meet \opttype{\type_2} = & \opttype{\type_1 \meet \type_2} \\ 
        \semtype{\type_1} \meet \semtype{\type_2} = & \semtype{\type_1 \sqcap \type_2} \\
        \semtype{\type} \meet {\sf CharSeq} = & \semtype{\type} \\\\
        
        \type \join {\sf Any} = & {\sf Any} \\
        \type_1 \join \opttype{\type_2} = & \opttype{\type_1 \join \type_2} \\ 
        \opttype{\type_1} \join \opttype{\type_2} = & \opttype{\type_1 \join \type_2} \\ 
        \semtype{\type_1} \join \semtype{\type_2} = & \semtype{\type_1 \sqcup \type_2} \\ 
        \semtype{\type} \join {\sf CharSeq} = & {\sf Any}
    \end{array}
    \]
    \vspace{-0.5cm}
    \caption{Type intersection and union.}
    \label{fig:type_intersection}
    \vspace{-0.5cm}
    \end{figure}

\section{Learning Semantic Regexes from Examples}\label{sec:sketch_synthesis}

In this section, we describe our synthesis algorithm for solving the semantic string matching problem from examples. Our method involves two main steps: generating a \emph{typed sketch} from the positive examples and completing the sketch using an enumerative search-based synthesizer. If sketch completion fails, our method refines the sketch and performs synthesis using the new sketch.
In the rest of this section, we first provide some preliminary information, then present our top-level learning algorithm, and then describe each of its key components.

\subsection{Sketch Language}\label{sec:sketch}

Our learning algorithm crucially relies on the notion of a \emph{typed sketch} whose syntax is shown in Figure~\ref{fig:sketch}. At a high level, the sketch language extends our semantic regex DSL by allowing a ``typed hole''  (denoted $\thole{\type}$) which represents an arbitrary expression  of type $\type$. Given a sketch $\sketch$, we use the notation $\sem{\sketch}$ to denote the set of all semantic regexes that can be obtained by completing holes in $\sketch$ by valid expressions of the corresponding type. Figure~\ref{fig:sketch} also defines sketch semantics in terms of the space of all programs they represent.

\begin{figure}[!t]
\vspace{-0.5cm}
    \small
    \begin{minipage}[c]{0.45\textwidth}
    \[
    \begin{array}{rll}
    \sketch := & \regex     & \text{(regex)} \\
            & | \ {\tt f}(\overline{\sketch}) & \text{(operator in the language)}\\ 
            & | \ \thole{\type} & \text{(typed hole)}
    \end{array}
    \]
    \end{minipage}
    \begin{minipage}[c]{0.45\textwidth}
    \[
        \begin{array}{rl}
        \sem{\regex} = & \{\regex\} \\ 
        \sem{{\tt f}(\overline{\sketch})} = & \{{\tt f}(\overline{\regex} \ | \ \forall_{i \in | \overline{\regex}| } \regex_i \in \sem{ \sketch_i}) \}\\
        \sem{\thole{\type}} = & \{\regex \ | \ \vdash \regex : \type\}
        \end{array}
    \]
    \end{minipage}
    \caption{Sketch syntax and its semantics. Here ${\tt f}$ refers to any construct in the DSL defined in Figure~\ref{fig:dsl}.}
    \label{fig:sketch}
    \end{figure}

\begin{example}
Consider the sketch $\thole{Organization} \cdot ``.com"$, which represents the space of semantic regexes that match strings consisting of an organization name followed by the string constant ``.com''. Possible completions of this sketch include, but are not limited, to the following semantic regexes: (1) $\matchsemq{{\sf Company}}{}\cdot ``.com"$, (2) $\matchsemq{{\sf Institution}}{} \cdot ``.com"$, and (3)  $(\matchsemq{{\sf Institution}}{} \cup \matchsemq{{\sf Company}}{}) \cdot ``.com"$.
% \gd{one thing I don't get is how we know Company is a subtype of organization.}
\end{example}

\subsection{Top-level algorithm}\label{sec:top-level}

Our top-level algorithm is outlined in Figure~\ref{fig:top-level}. Given a set of positive examples $\ex^+$ and a set of negative examples $\ex^-$, \textsc{Synthesize} returns a semantic regex that accepts all positive examples and rejects all negative examples. At a high level, the algorithm  repeatedly generates a new sketch using a large language model, then attempts to find a valid instantiation of that sketch, and continues this process  until it finds a regex that is consistent with all user-provided examples. Intuitively, each candidate sketch serves as a possible generalization of the positive examples, and the goal of the synthesizer is to determine whether that sketch is a suitable generalization. 

\begin{figure}[t]
\vspace{-0.5cm}
    \small
    \begin{algorithm}[H]
    \begin{algorithmic}[1]
    \Procedure{Synthesize}{$\ex^+, \ex^-$}
    \Statex \Input{A set of positive $\ex^+$ and negative examples $\ex^-$.}
    \Statex \Output{A program that is consistent with the examples.}
    \State $\sketch_f \assign \bot$;
    \While{${\tt HasMoreSketch}(\ex^+)$}
    \State $\sketch \assign \textsc{GetNextSketch}(\ex^+, \sketch_f)$;
    \While{${\tt HasDecomp}(\sketch, \ex^+)$}
    \State $\goal \assign \textsc{GetNextDecomp}(\sketch, \ex^+)$;
    \State $M \assign \textsc{SynthesizeFromDecomp}(\sketch, \goal, \ex^-)$;
    \If{$M \neq \bot$} \Return $\sketch[M]$;
    \EndIf 
    \EndWhile
    \State $\sketch_f \assign \sketch$;
    \EndWhile
    \State \Return $\bot$;
    \EndProcedure 
    \end{algorithmic}
    \end{algorithm}
    \vspace{-0.5in}
    \caption{Top-level synthesis algorithm. Here, $\sketch[M]$ means replacing each hole $h \in \sketch$ with $M[h]$. } 
    \label{fig:top-level}
    \vspace{-0.5cm}
\end{figure}

In more detail, the \textsc{Synthesize} procedure first calls  \textsc{GetNextSketch}, which queries GPT-3 to produce a sketch $\sketch$ that is likely to satisfy the positive examples.  Then, for a given  sketch $\sketch$, {\sc GetNextDecomp}  infers a \emph{decomposition} $\goal$, which is a mapping from each hole in $ \sketch$ to a set of positive examples for that hole.  Then, for a given decomposition $\goal$, the algorithm calls {\sc SynthesizeFromDecomp} to perform compositional synthesis based on the inferred specification $\goal$. 

If the call to {\sc SynthesizeFromDecomp} returns a non-empty mapping $M$, which maps each hole in $\sketch$ to a concrete regex $\regex$, we find a solution that is consistent with the specification and returns the synthesized regex by replacing the holes in $\sketch$ with the corresponding solution in $M$. Otherwise, if the call to {\sc SynthesizeFromDecomp} yields $\bot$, there are two possibilities: Either the decomposition $\goal$ is incorrect (recall from Section~\ref{sec:overview} that there is ambiguity in how to assign positive examples to holes), or the sketch $\sketch$ itself is incorrect. In the former case, the algorithm considers a different decomposition, which maps at least one of the holes in the sketch to a different set of examples. If the algorithm exhausts all possible decompositions, this means that the sketch must be incorrect and the algorithm repairs the current sketch by performing fault localization and querying GPT-3 to produce a different generalization of the positive examples. This process continues until the algorithm finds a globally consistent regex with all (positive \emph{and} negative) examples or runs out of possible sketches. 
 In the following discussion, we explain each of the three components (decomposition, type-directed synthesis, and sketch repair) in more detail.
%Given a sketch $\sketch$, along with the positive examples $\ex^+$ and negative examples $\ex^-$, the goal is to either instantiate the sketch to $\bot$ that indicates an unsuccessful synthesis attempt or a program in our language $\regex$ such that:
%\begin{equation*}
%    \small
%   (1) \ \regex \in \sem{\sketch} \ \ \ \ (2) \  \forall_{e \in \ex^+}. \ {\tt match}(e, \regex) \ \ \ \  (3) \ \forall_{e \in \ex^-}. \ \neg {\tt match}(e, \regex)
%\end{equation*}

\subsection{Decomposing the Specification}\label{sec:decomp}
To perform compositional synthesis, our learning algorithm decomposes the global specification into a \emph{set} of specifications, one for each hole in the sketch. In this section, we describe the {\sc GetNextDecomp} procedure for specification decomposition using the inference rules  in Figure~\ref{fig:get_next_goal}, which derive judgments of the following shape:
\[
\ex^+\vdash \sketch \leadsto \goal
\]
The meaning of this judgment is that, given positive examples $\ex^+$, $\goal$ is a \emph{possible} decomposition  that maps each hole in the sketch to its corresponding positive  examples. As mentioned earlier, the decomposition is, in general, \emph{not} unique, so there can be multiple decompositions $\goal_1, \ldots, \goal_n$ for a given sketch $\sketch$. 

We now explain the decomposition rules from Figure~\ref{fig:get_next_goal} in more detail. The first rule, labeled {\sc Sketch-Match}, considers a program sketch with top-level operator $f$ (e.g., concatenation or intersection) and sub-sketches $S_1, \ldots, S_n$. To infer a specification for each hole in $\sketch$, we first generate a regex $\regex^\star$ that over-approximates $\sketch$ (via the call to 
${\tt OverApprox}$). Intuitively, ${\tt OverApprox}$ generates a regex $\regex^\star$ such that for \emph{any} $r \in \sem{S}$, $\regex^\star$ accepts every string that is accepted by $r$. Because our over-approximation approach is exactly the same as used in prior work~\cite{alpharegex,regel}, we do not formally present it, but the basic idea is  to replace each hole that appears under an even (resp. odd) number of negation symbols by the regex $.*$ (resp. $\emptyset$).  This method guarantees that the resulting regex $r^\star$ will accept every string that is accepted by any instantiation of $\sketch$. Furthermore, note that $r^\star$ is a standard regex without any semantic pattern matching constructs, as all holes have been replaced by either the universal or the empty set. 

\begin{figure}[t]
\vspace{-0.5cm}
    \small
    \begin{mathpar}
       \inferrule*[left=Sketch-Match]{
       \sketch = f(\sketch_1, \ldots, \sketch_n) \ \ \ r ^\star= {\tt OverApprox}(\sketch)  \\\\
       (\ex_1^+, \ldots, \ex_n^+) \in {\tt Match}(r^\star, \ex^+) \\\\
        \ex^{+}_i \vdash \sketch_i \leadsto \goal_i  \ \ \ i \in \{1, \ldots, n \}
        }{\ex^+ \vdash \sketch \leadsto {\tt Merge}(\goal_1, \ldots, \goal_n)}\\
        \inferrule*[left=Sketch-NoPosMatch]{
       \sketch = f(\sketch_1, \ldots, \sketch_n) \ \ \ r ^\star= {\tt OverApprox}(\sketch)  \\\\
       {\tt Match}(r^\star, \ex^+) \equiv \emptyset
        }{\ex^+ \vdash \sketch \leadsto \bot}\\
        \inferrule*[left=Concrete-Feasible]{
        {\tt Match}(r, \ex^+) \neq \emptyset
        }{\ex^+ \vdash \regex \leadsto \emptyset}\and
        \inferrule*[left=Concrete-Infeasible]{
        {\tt Match}(r, \ex^+) \equiv \emptyset
        }{\ex^+ \vdash \regex \leadsto \bot} \\
       \inferrule*[left=Hole-Feasible]{\forall e \in \ex^+. \ \ \ {\tt SemanticType}(e) \subtype  \type
       }{\ex^+ \vdash \thole{\type} \leadsto [\thole{\type} \mapsto \ex^+]}\and       
       \inferrule*[left=Hole-Infeasible]{\exists e \in \ex^+. \ \ \ {\tt SemanticType}(e) \not\subtype \type}{\ex^+ \vdash \thole{\type} \leadsto \bot}
    \end{mathpar}
    \caption{Procedure for $\textsc{GetNextDecomp}(\sketch, \ex^+)$. ${\tt OverApprox}(\sketch)$ returns a concrete regex that over-approximates $\sketch$. ${\tt Merge}$ returns $\bot$ if one of its argument is $\bot$, otherwise it disjointly unions all its arguments.}
    \label{fig:get_next_goal}
    \vspace{-0.5cm}
\end{figure}

Next, once we generate the over-approximation $\regex^\star$, we infer positive  examples for each sub-sketch $\sketch_1, \ldots, \sketch_n$ used in $\sketch$. To do so, for each positive example $e$, we use a standard regex matching tool to find a parse of $e$ into the format $f(S_1, \ldots, S_n)$ with corresponding sub-strings $e_i$ for each sub-sketch $\sketch_i$. After propagating each example $e_i$ to nested sketch $S_i$ and recursively applying the inference rules, we obtain the decomposed specifications $\goal_1, \ldots, \goal_n$ for each of the sub-sketches in $\sketch$. These mappings are finally combined via the call to the ${\tt Merge}$ function, defined as follows:
\[
\small
{\tt Merge}(\goal_1, \ldots, \goal_n) = 
\left \{
\begin{array}{ll}
\bot & {\rm if} \ \exists i \in [1, n]. \ \goal_i = \bot \\
\biguplus_{i=1}^n \goal_i & {\rm otherwise}
 \end{array}
\right .
\]
where the notation $\uplus$ indicates disjoint union.

The next rule, labeled {\sc Sketch-NoMatch}, corresponds to an infeasible sketch or decomposition. Because every string accepted by  $r \in \sem{\sketch}$ must also be 
accepted by the over-approximation $r^\star$, the algorithm yields $\bot$ to indicate a failure when $r^\star$ doesn't match at least one of the positive examples.

The remaining rules correspond to the base cases of the recursive decomposition algorithm. Specifically, the rules prefixed with {\sc Concrete} consider the case where the sketch is a concrete regex $r$ without a hole. Specifically, we check the feasibility of $r$ by testing whether it matches all of the positive examples. If so, the sketch is feasible, and the algorithm returns the empty mapping $\emptyset$. Otherwise (the {\sc Concrete-Infeasible} case), the algorithm returns $\bot$ to indicate failure. 

The final two rules correspond to base cases for a  hole and utilize the fact that sketches are typed. In particular, given a hole of type $\tau$, if there exists a positive example $e \in \mathcal{E}^+$ whose type is not $\tau$, this indicates a conflict and the algorithm returns $\bot$ in the {\sc Hole-Infeasible} rule. Otherwise, in the {\sc Hole-Feasible} rule, the constructed specification maps this hole to the input positive examples $\mathcal{E}^+$.

\begin{example}
   Consider the positive examples from Section~\ref{sec:overview} and  the following sketch: \[\small \thole{Name}\cdot ``, \ " \cdot\thole{Country}\cdot ``, \ " \cdot\thole{Year}_1\cdot ``-" \cdot\thole{Year}_2 \]
   The over-approximation for this sketch is the following regex:
\[\small
.* \cdot ``, \ " \cdot .* \cdot ``, \ " \cdot .*  \cdot ``-" \cdot .* \]
Using our decomposition technique, we infer the following positive examples for each hole:
    \begin{center}
    \footnotesize
    \begin{tabular}{cccc}
        \toprule
         $\thole{Name}$ & $\thole{Country}$ & $\thole{Year}_1$ & $\thole{Year}_2$  \\
         \midrule
         John Thomas Young Gilroy & Britain & 1898 & 1985 \\ 
         Thomas Hudson & Britain & 1701 & 1779 \\
         Thomas Couture & France & 1815 & 1879 \\
         \bottomrule
    \end{tabular} 
\end{center}
\end{example}

We conclude this subsection by stating the theorem about the soundness of decomposition:
\begin{theorem}
Consider the synthesis problem with positive examples $\mathcal{E}^+$. Let $\sketch$ be a candidate sketch and let $r$ be a completion of $\sketch$ mapping each hole $ h_i$ in $\sketch$ to a semantic regex $r_i$. If $r$ satisfies all positive examples $\mathcal{E}^+$, then there exists some $\goal \in {\textsc{GetNextDecomp}}(\sketch, \mathcal{E}^+)$ such that every $r_i$ satisfies $\goal[h_i]$.

%let $\psi_1, \ldots, \psi_n$ be the set of all decompositions generated by our technique (according to the rules of Figure~\ref{fig:get_next_goal}. If the sketch $\sketch$ is feasible with completion $r$  
\end{theorem}

\subsection{Compositional Type-Directed Synthesis}\label{sec:compsitional}

Next, we explain our compositional learning technique for synthesizing a semantic regex for a given sketch and decomposed specification. This algorithm, called {\sc SynthesizeFromDecomp}, is shown in Figure~\ref{fig:goal_synthesis}. Given a sketch $\sketch$, specification $\goal$, and negative examples $\mathcal{E}^-$, the {recursive} {\sc SynthesizeFromDecomp} procedure lazily generates possible sketch completions until it finds a regex that is globally consistent with the top-level specification. %Because the decomposed specification already guarantees consistency with the positive examples, a candidate completion can only violate the top-level spec

\begin{figure}[t]
\small
\vspace{-0.5cm}
    \begin{algorithm}[H]
       \begin{algorithmic}[1]
       \Procedure{SynthesizeFromDecomp}{$\sketch, \goal, \ex^-$} 
       \Statex \Input{A sketch $\sketch$, a specification $\goal$, a set of negative examples $\ex^-$.}
       \Statex \Output{A sketch completion consistent with all examples.}
       \State $\holesym \assign {\tt ChooseHole}(\sketch)$
       \While{${\tt True}$}
       \State $\regex \assign \textsc{GetNextCompletion}({\tt TypeOf}(\holesym), \goal[\holesym], {\tt GetAllSubstr}(\ex^-))$;
       \If{$\regex \equiv \bot$} \Return $\bot$;
        \EndIf
        \While{{\tt True}}
        \State $M \assign \textsc{SynthesizeFromDecomp}(\sketch[\regex/\holesym], \goal, \ex^-)$
        \If{$M \equiv \bot$} {\bf break};
        \EndIf
        \State $M' \assign M \cup [\holesym \mapsto \regex]$; %$\regex' \assign \sketch[M']$;
        \If{${\tt Reject}(\sketch[M'], \ex^-)$}  {\bf yield} \  $M'$;
        \EndIf
        \EndWhile
       \EndWhile
       \State \Return $\bot$;
       \EndProcedure
       \end{algorithmic}
   \end{algorithm}
   \vspace{-0.5in}
   \caption{Sketch completion algorithm for a given decomposition.}
   \label{fig:goal_synthesis}
   \vspace{-0.5cm}
\end{figure}

To perform synthesis for a given specification, the algorithm starts by choosing one of the holes $h$ in the sketch (line 2) and synthesizes a completion $r$ for that hole \emph{only} by calling {\sc GetNextCompletion} at line 4. Then, the loop in lines 6--10 tries to find a completion for the remaining holes. In particular, in each iteration of the nested loop, the algorithm recursively calls {\sc SynthesizeFromDecomp} to fill all remaining holes, assuming that $h$ is replaced by $r$. If synthesis fails (i.e., $M \equiv \bot$ at line 8), the algorithm moves on to a different completion of $h$. Otherwise, it checks if the current solution (which is obtained by instantiating $\sketch$ with $M \cup [ h \mapsto r] $) rejects all negative examples, and if so, returns this solution. %Otherwise, in the next iteration, it considers a different completion of the sketch through the recursive call to {\sc SynthesizeFromDecomp}. 

\begin{figure}[t]
\vspace{-0.5cm}
\small
    \begin{algorithm}[H]
    \begin{algorithmic}[1]
        \Procedure{GetNextCompletion}{$\goaltype, \ex^+, \ex_{\star}$}
        \Statex \Input{Goal type $\goaltype$,  positive examples $\ex^+$, and a set of strings $\ex_{\star}$ for checking observational equivalence.}
        \Statex \Output{A program of type $\goaltype$ that matches $\ex^+$}
        \State $\prog_0 \assign ((s_\grammar, \goaltype), \emptyset)$; 
        \State $\worklist \assign \{\prog_0\}$; $\res \assign \{\}$;
        \While{$\worklist \neq \emptyset$}
            \State $\prog \assign \worklist.remove()$;
            \If{$\texttt{IsComplete}(\prog)$}
                \If{$\vdash \prog: \goaltype \wedge \bigwedge_{e \in \ex^+} {\tt match}(e, \prog)$}
                \State $\ex \assign \{e \ | \ e \in \ex_{\star} \wedge \neg {\tt match}(e, \prog)\}$;
                \If{$\ex \notin \res$} $\res \assign \res \cup \{\ex\}$; {\bf yield} $\prog$;
                \EndIf
                \EndIf
            \Else
                \ForAll{$\prog' \in \texttt{Expand}(\prog)$}
                \If{$\exists n \in {\tt Nodes}(\prog') . \ {\tt IsComplete}(\prog'(n)) \wedge \vdash {\tt TypeOf}(\prog'(n)) \not\subtype {\tt GoalType}(n)$} 
                \State {\bf continue};
                \ElsIf{$\exists e \in \ex^+. \neg {\tt match}(e, {\tt OverApprox}(\prog'))$} 
                \State {\bf continue};
                \EndIf
                \State $\worklist \assign \worklist \cup \{\prog'\}$;
                \EndFor
            \EndIf
        \EndWhile
        \State \Return $\bot$;
        \EndProcedure    
    \end{algorithmic}    
    \end{algorithm}
    \vspace{-0.6cm}
    \caption{Hole synthesis algorithm. {\tt OverApprox} follows the procedure as described in {\sc Regel}~\cite{regel}. }
    \label{fig:hole_synthesis}
    \vspace{-0.5cm}
    \end{figure}

The final missing piece for our sketch instantiation algorithm is the \textsc{GetNextCompletion} procedure shown in Figure~\ref{fig:hole_synthesis} which performs synthesis for a \emph{single} hole. At a high level, this algorithm performs top-down enumerative search and uses a combination of types~\cite{synquid,myth, frankle16} and observational equivalence~\cite{observational-equiv} to prune the search space. As standard in top-down search, this algorithm utilizes the notion of  \emph{partial programs}~\cite{neo,morpheus}, which can be thought of as an abstract-syntax tree where some of the nodes are labeled with non-terminals to be expanded later. 

In more detail, the hole synthesis algorithm utilizes a worklist $\mathcal{W}$, which is initialized to a partial program $P_0$ with a single node (lines 2--3). Each node in the partial program is annotated with a grammar symbol (in this case, the start symbol $s_\grammar$) and its corresponding type (in this case, $\tau_h$). Then, in each iteration of the loop in lines 4--16, the algorithm dequeues one of the partial programs $P$ in the worklist and processes it. If the partial program is complete (meaning that all nodes are labeled with terminal symbols), the algorithm performs the following checks:
\begin{enumerate}[leftmargin=*]
\item {\bf Type consistency:} If the type of $P$ is \emph{not} $\tau_h$, $P$ clearly does not have the intended type and is rejected (line 7).
\item {\bf Consistency with examples:} If $P$ does not satisfy all positive examples $\mathcal{E}^+$, it does not satisfy the specification and is also rejected at line 7.
\item {\bf Observational equivalence:} If $P$ rejects the \emph{exact same set} of strings as a program the algorithm has previously encountered, it is redundant to consider $P$, as it is observationally equivalent to another solution $P'$ that has been rejected. Hence, the algorithm only yields $P$ as a solution if it is observationally different from a previously encountered solution (lines 8--9).
\end{enumerate}

On the other hand, if the current partial program $P$ is \emph{incomplete} (meaning it has at least one ``open'' node labeled with a non-terminal), the algorithm chooses one of the open nodes  and expands it using the available productions in the grammar (line 11). In particular, given an open node $n$ labeled with a non-terminal $N$, the {\tt Expand} procedure considers each production of the form $ N \rightarrow \alpha$ and adds new nodes where  each new node with a grammar symbol and its corresponding (inferred) type. However, because a resulting expansion $P'$ may not necessarily be feasible, the algorithm performs two additional checks before adding $P'$ to the worklist at line 16:

\begin{itemize}[leftmargin=*]
    \item {\bf Type-directed feasibility check:} For each complete subprogram $P_i$ of $P'$, the algorithm checks if the actual type of $P_i$  is a subtype of its annotated goal type (line 12). If this type feasibility check fails for \emph{any} node $n$, then program $P'$ is pruned from the search space, and none of its expansions are considered.
    \item {\bf Feasibility check using over-approximation:} Additionally, the algorithm constructs an over-approximating regular expression $r^\star$ that accepts every string that is accepted by any $r \in \sem{P'}$ using the same {\tt OverApprox} procedure from Section~\ref{sec:decomp}. If this over-approximation $r^\star$ fails to match one of the positive examples, $P'$ is infeasible and therefore pruned away at lines 14--15. 
\end{itemize}

Otherwise, $P'$ is added to the worklist, and the search process continues until a solution is found.

\begin{theorem}
Let $R$ be the set of solutions returned by {\sc GetNextCompletion}$(\tau_h, \mathcal{E}^+, \mathcal{E}_\star)$. We have:
\begin{itemize}[leftmargin=*]
    \item {\bf Soundness:} Every $r \in R$ is a solution to the hole synthesis problem, meaning (1) $r$ has type $\tau_h$ and (2) satisfies examples $\mathcal{E}^+$
    \item {\bf Completeness:} If $r \not \in R$, then $r$ is either not a solution or is observationally equivalent to some $r' \in R$ for strings $\mathcal{E}_\star$. 
\end{itemize}
\end{theorem}

\subsection{Sketch Generation}\label{sec:sketch_gen}

In the final part of this section, we describe our technique for generating typed sketches from examples. In particular, we employ few-shot prompting and build our sketch generator on top of GPT-3~\cite{gpt3}.
% \gd{you should pick either GPT-3 or GPT3 and stick to it. also maybe use 3.5} 
%As mentioned in the previous section, we choose GPT3 over other models because of its ability to reason about semantics and learn from \emph{very few} examples.

\subsubsection{Background on Few-Shot Prompting with LLMs} In recent years, large language models (LLMs)~\cite{gpt3, palm} have made major breakthroughs in natural language understanding. These are models $P(\mathbf{x}) = P(x_1) P(x_2 \mid x_1) \ldots P(x_n \mid x_1, \ldots, x_{n-1})$ modeling a sequence as a product of distributions over each next word via the chain rule.
% \gd{feel free to cut this}

By showing LLMs a few examples of a task to perform and then giving them a test example, LLMs can perform that task on the test example via \emph{in-context learning}, without retraining or fine-tuning the model's parameters. The user only needs to provide a few examples and invoke the model's next-word prediction capabilities (repeatedly taking the most likely next token under the model).
%., and the model can often achieve satisfactory performance on the new task.
To give a concrete example, consider the task of transforming numbers in strings to texts, a task that GPT-3 has not specifically been  trained on. Figure~\ref{fig:example_prompt} shows a typical usage scenario of GPT-3 when performing such a task: here, line 1 provides the task description, lines 2-4 provides a few examples,  line 5 is the query, and the output of the model is highlighted in red. 
%The red box highlighted the output of the model, ``fromage'', which is indeed the french translation of ``cheese''. Noted that during the entire process, no gradient updates are performed.

\begin{figure}
    \vspace{-0.2cm}
    \centering
    \includegraphics[scale=0.3, trim=0 50 800 820, clip]{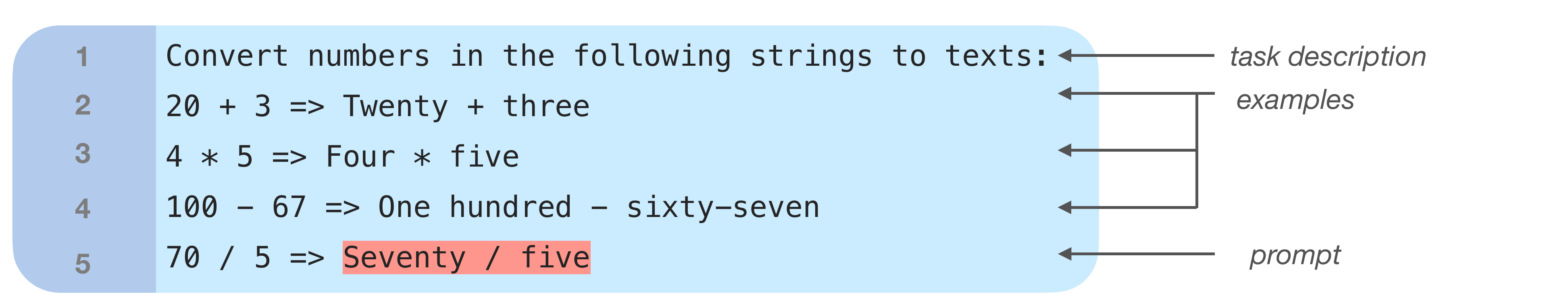}
    \vspace{-0.2cm}
    \caption{Sample input for a few-shot string transformation to GPT-3 and its output is highlighted in red.}
    \label{fig:example_prompt}
    \vspace{-0.5cm}
\end{figure}

\begin{figure}[t]
\small
\vspace{-0.5cm}
    \begin{algorithm}[H]
    \begin{algorithmic}[1]
        \Procedure{GetNextSketch}{$\ex^+, \sketch_f$}
        \Statex \Input{A set of positive examples $\ex^+$, and an optional infeasible sketch $\sketch_f$.}
        \Statex \Output{A new sketch that has not been generated so far.}
        \State $\sketch_{All} \assign \emptyset$
        \While{${\tt True}$}
            \If{$\sketch_f \equiv \bot$}
            \State $\sketch \assign {\tt GetSketch}(\ex^+)$;
            \If{$\sketch \not\in \sketch_{All}$} $\sketch_{All} \assign \sketch_{All} \cup \{\sketch\}$; ${\bf yield} \ \sketch$;
            \EndIf
            \Else
            \While{${\tt HasRepair}(\sketch_f, \ex^+)$}
            \State $\metasketch, \goal \assign \textsc{LocateError}(\sketch_f, \ex^+)$;
            \State $\sketch \assign \metasketch[\holesym_i \mapsto {\tt GetSketch}(\goal[\holesym_i]) \ | \ h_i \in \mathsf{MetaHoles}(\metasketch)]$
            \If{$\sketch \not\in \sketch_{All}$} $\sketch_{All} \assign \sketch_{All} \cup \{\sketch\}$; ${\bf yield} \ \sketch$;
            \EndIf
            \EndWhile
            \EndIf
        \EndWhile
        \State \Return $\bot$;
        \EndProcedure 
    \end{algorithmic} 
    \end{algorithm}
    \vspace{-0.7cm}
    \caption{Sketch generation procedure. ${\tt GetSketch}(\ex^+)$ prompts the neural model for a new sketch, as illustrated in Figure~\ref{fig:siren_prompt}.}
    \label{fig:sketch_gen}
    \vspace{-0.5cm}
\end{figure}

\begin{figure}
    \centering
    \includegraphics[scale=0.22, trim=0 20 0 400, clip]{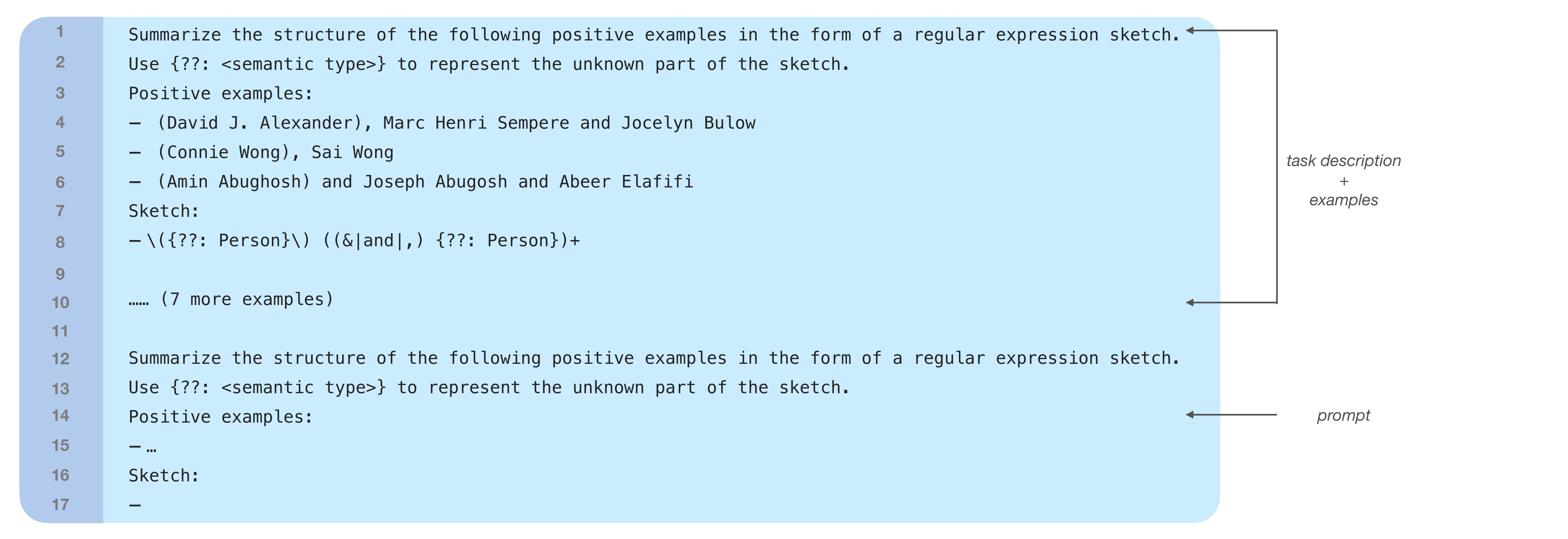}
    \vspace{-0.6cm}
    \caption{GPT-3 input structure for generating a sketch for the semantic string matching task.}
    \label{fig:siren_prompt}
    \vspace{-0.5cm}
\end{figure}

\subsubsection{Querying LLM for Sketches}
To obtain typed sketches, our approach prompts GPT-3 with suitable queries.\footnote{We consider sketches generated from models \texttt{text-davinci-003}, \texttt{code-davinci-002} and \texttt{gpt-3.5-turbo}.}  As shown in Figure~\ref{fig:sketch_gen}, the {\sc GetNextSketch} procedure takes as input positive examples $\mathcal{E}^+$ and an optional infeasible sketch $\sketch_f$, which is used in later iterations of the  algorithm for  sketch repair. 
Initially, the algorithm starts by querying GPT-3 for a sketch using the {\tt GetSketch} procedure, as illustrated in Figure~\ref{fig:siren_prompt}. The prompt to GPT-3 contains a task description, a manually-curated set of representative examples (in the form of a query and its desired output), and, finally, the prompt itself (lines 12--17 in Figure~\ref{fig:siren_prompt}). The {\tt GetSketch} procedure then attempts to parse the model's output into a typed sketch; however, there is no guarantee that the GPT-3 output will belong to our sketch grammar. Hence, if parsing fails, the {\tt GetSketch} procedure keeps prompting GPT-3 for a new sketch until the model's output is parseable.\footnote{Past work has explored few-shot semantic parsing from natural language into DSLs using structured natural language as an intermediate representation \cite{shin-van-durme-2022-shot}; however, more recent work has shown that LLMs can do well at this task without such guidance, even in the presence of adversarial perturbations \cite{zhuo2023robustness}.}

In future invocations of  \textsc{GetNextSketch}, this procedure may be invoked with an infeasible sketch $\sketch_f$ that needs to be repaired. Lines 8--11 of Figure~\ref{fig:sketch_gen} deal with this sketch repair aspect of the algorithm. Specifically, given the infeasible sketch $\sketch_f$ and positive examples $\mathcal{E}^+$, {\sc LocateError} produces a \emph{repair specification}, which consists of a so-called \emph{meta-sketch} $\metasketch$ and a specification $\goal$. A meta-sketch is like a sketch except that it contains \emph{untyped} ``meta-holes'' that need to be instantiated with a \emph{typed sketch}. The specification $\goal$ maps each meta-hole in $\metasketch$ to a set of positive examples.  Such a meta-sketch is instantiated into a regular sketch by querying GPT-3 via the {\tt GetSketch} procedure for each of the meta-holes $h_i$ in $\metasketch$ and its corresponding examples $\Psi[h_i]$.

Finally, we turn our attention to the {\sc LocateError} procedure, which is presented as inference rules in Figure~\ref{fig:sketch_repair}. These rules  derive judgments of the following shape:
\[
\ex^+ \vdash \sketch \hookrightarrow \metasketch, \goal
\]
meaning that $(\metasketch, \goal)$ is a repair specification for infeasible sketch $\sketch$ and examples $\mathcal{E}^+$.
The fault localization rules in Figure~\ref{fig:sketch_repair}  largely resemble \textsc{GetNextDecomp} for performing decomposition in that they use over-approximations. We explain these rules in more detail below.

\begin{figure}
    \centering
    \vspace{-0.5cm}
    \footnotesize
    \begin{mathpar}
        \inferrule*[leftstyle={\footnotesize \sc}, Left=Sketch-Single-Fail]{ {\tt Match}({\tt OverApprox}(\sketch), \ex^+) \equiv \emptyset \\\\
        \sketch = f(\sketch_1, \ldots, \sketch_n) \\\\
        \exists i \in \{1, \ldots, n\} \ \ \  \regex_i^\star = {\tt OverApprox}(\sketch[\hole/\sketch_i]) \\\\
        (\ex^+_1, \ldots, \ex^+_i, \ldots, \ex^+_n) \in {\tt Match}(\regex_i^\star, \ex^+) \\\\
        \ex^+_i \vdash \sketch_i \hookrightarrow \metasketch_i, \goal_i
        }{\ex^+ \vdash \sketch \hookrightarrow \sketch[\metasketch_i/\sketch_i], \goal_i} \\
        \inferrule*[leftstyle={\footnotesize \sc}, Left=Sketch-Multi-Fail]{ {\tt Match}({\tt OverApprox}(\sketch), \ex^+) \equiv \emptyset \\\\
        \sketch = f(\sketch_1, \ldots, \sketch_n) \\\\
        \forall i \in \{1, \ldots, n\} \ \ \  \regex_i^\star = {\tt OverApprox}(\sketch[\hole/\sketch_i]) \\\\
        {\tt Match}(\regex_i^\star, \ex^+) \equiv \emptyset
        }{\ex^+ \vdash \sketch \hookrightarrow \hole, [\hole \mapsto \ex^+]}\\
         \inferrule*[leftstyle={\footnotesize \sc}, Left=Sketch-Nested-Fail]{
        \sketch = f(\sketch_1, \ldots, \sketch_n) \\\\
        (\ex^+_1, \ldots, \ex^+_n) \in {\tt Match}({\tt OverApprox}(\sketch), \ex^+) \\\\
        \ex^+_i \vdash \sketch_i \hookrightarrow \metasketch_i, \goal_i \ \ \ i \in \{1, \ldots, n\}
        }{\ex^+ \vdash \sketch \hookrightarrow f(\metasketch_1, \ldots, \metasketch_n), {\tt Merge}(\goal_1, \ldots, \goal_n)}\\
        \inferrule*[leftstyle={\footnotesize \sc}, Left=Hole-Fail]{\exists e \in \ex^+. \ {\tt SemanticType}(e) \not\subtype \tau}{\ex^+ \vdash \thole{\type} \hookrightarrow \hole, [\hole \mapsto \ex^+]}\and\and
        \inferrule*[leftstyle={\footnotesize \sc}, Left=Hole-Correct]{\forall e \in \ex^+. \ {\tt SemanticType}(e) \subtype \tau}{\ex^+ \vdash \thole{\type} \hookrightarrow \thole{\type}, \emptyset}\\
        \inferrule*[leftstyle={\footnotesize \sc}, Left=Concrete-Fail]{{\tt Match}(r, \ex^+) \equiv \emptyset}{\ex^+ \vdash r \hookrightarrow \hole, [\hole \mapsto \ex^+]}\and\and
        \inferrule*[leftstyle={\footnotesize \sc}, Left=Concrete-Correct]{{\tt Match}(r, \ex^+) \neq \emptyset}{\ex^+ \vdash r \hookrightarrow r, \emptyset}     
    \end{mathpar}
    \caption{Procedure for \textsc{LocateError}.}
    \label{fig:sketch_repair}
    \vspace{-0.5cm}
\end{figure}

\paragraph{{ \sc \bf Sketch-Single-Fail.}} This rule applies to a sketch $\sketch$ of the form $f(\sketch_1, \ldots, \sketch_n)$ where (1) there is at least one positive example that is not matched by the over-approximation of $\sketch$ (premise on the first line) and (2) where only one of the sub-sketches $\sketch_i$ is faulty. To determine whether condition (2) holds, this rule replaces the entire sub-sketch $S_i$ with a single hole and then checks whether the over-approximation of the resulting sketch can accept all positive examples. If so, it recursively performs fault localization on $S_i$ and returns a meta-sketch by replacing $S_i$ in $S$ with its corresponding meta-sketch $\metasketch_i$. 

\paragraph{{ \sc \bf Sketch-Multi-Fail.}} This rule is similar to the first one except that it deals with the scenario where there are multiple faulty sub-sketches. That is, even after we replace any individual sub-sketch with a hole, there is \emph{still} at least one positive example that is not matched by the over-approximation. In this case, we generate a meta-sketch that consists of a single hole.

\paragraph{{ \sc \bf Sketch-Nested-Fail.}} This rule also applies to a sketch $\sketch$ of the form $f(\sketch_1, \ldots, \sketch_n)$ but considers the case where the over-approximation of $\sketch$ matches all the positive examples. However, as the sketch is infeasible, there must nonetheless be at least one problem inside the next sub-sketches. Hence, our fault localization technique recursively localizes the error in the sub-sketches and returns the merged result.

\paragraph{{{ \sc \bf Hole-Fail.}}} This rule applies to the case where the type of a hole is incorrect in that its annotated type is inconsistent with at least one of the positive examples. In this case, our algorithm generates a meta-sketch by erasing the type annotation of this hole. 

\paragraph{{{ \sc \bf Hole-Correct, Concrete-Correct.}}} Since these rules apply to base cases without any problems, fault localization returns the original sketch.

\paragraph{{{ \sc \bf Concrete-Fail.}}} This rule applies to the case where a concrete regex does not match at least one of the examples. In this case, we simply replace the concrete regex with a meta-hole.

\begin{example}
    % Given the positive examples ``Mr. Hector'', ``Ms. Diego'' and ``Sir Leann'', we generate the following initial sketch:
    % \[
    % \thole{{\tt Title}} \cdot ``\ " \cdot  \thole{{\tt First\ name}} \cdot ``\ " \cdot \thole{{\tt Last\ name}}.
    % \]
    % Since the over-approximation of this sketch fails to match all the positive examples, we repair this sketch using the rule \textsc{Sketch-Single-Fail}. The first hole, $\thole{{\tt Title}}$ is correct because the first word of the positive examples do have the semantic type \texttt{Title},  so we keep this typed hole and recursively repair the sub-program $\cdot ``\ " \cdot  \thole{{\tt First\ name}} \cdot ``\ " \cdot \thole{{\tt Last\ name}}$. 
    Consider the positive examples from Section~\ref{sec:overview} and the following sketch:
    \[\small \thole{Name}\cdot ``, \ " \cdot\thole{Country}\cdot ``, \ " \cdot\thole{Year}  \]

    Suppose the synthesizer concluded this sketch to be infeasible since the string ``1898-1985'' cannot be identified as a year and sends this as a failed sketch to the sketch generator. To repair this sketch, we follow the \textsc{Sketch-Nested-Fail} rule to recursively traverse through each part of the sketch until we locate the faulty hole, $\thole{Year}$. We then gather the positive examples that should be matched by this hole, which are ``1898-1985'', ``1701-1779'' and ``1815-1879'', and replace the faulty typed hole with a new hole with no type (rule \textsc{Hole-Fail}). With the generated repair specification, we query GPT-3 to generate a new sketch for the faulty hole, and it returns a new sketch $\thole{Year} \cdot ``-" \cdot \thole{Year}$.
    
\end{example}

\section{Implementation}

We have implemented our synthesis algorithm in a new tool called \toolname written in Python. In this section, we provide implementation details about different components of \toolname.

\paragraph{{\bf Implementation of the semantic matching construct}} Our tool heavily relies on the use of GPT-3 to identify the semantic meanings of strings.\footnote{We use the \texttt{text-davinci-003} model.}
%Since GPT-3 is a general-purpose tool, we need to provide guidance to the model to learn our specific task via in-context learning.
Our few-shot prompt (following the discussion in Section~\ref{sec:sketch_gen}) to accomplish this is shown in Figure~\ref{fig:semantic_prompt}. The input begins with a task description that asks the model to identify \emph{all possible substrings} of a particular semantic type, and we instruct the model to return ``none'' if it does not find any. Following the task descriptions, we provide 8 examples,\footnote{We provide all the in-context examples we use in the supplementary material.} each of which shows the structure of a query: the first line provides the string of interest, and the second line specifies the semantic type of interest. Furthermore, we provide sample outputs for each example in the expected output format. 

\begin{figure}
    \centering
    \vspace{-0.5cm}
    \includegraphics[scale=0.25, trim=0 20 400 500, clip]{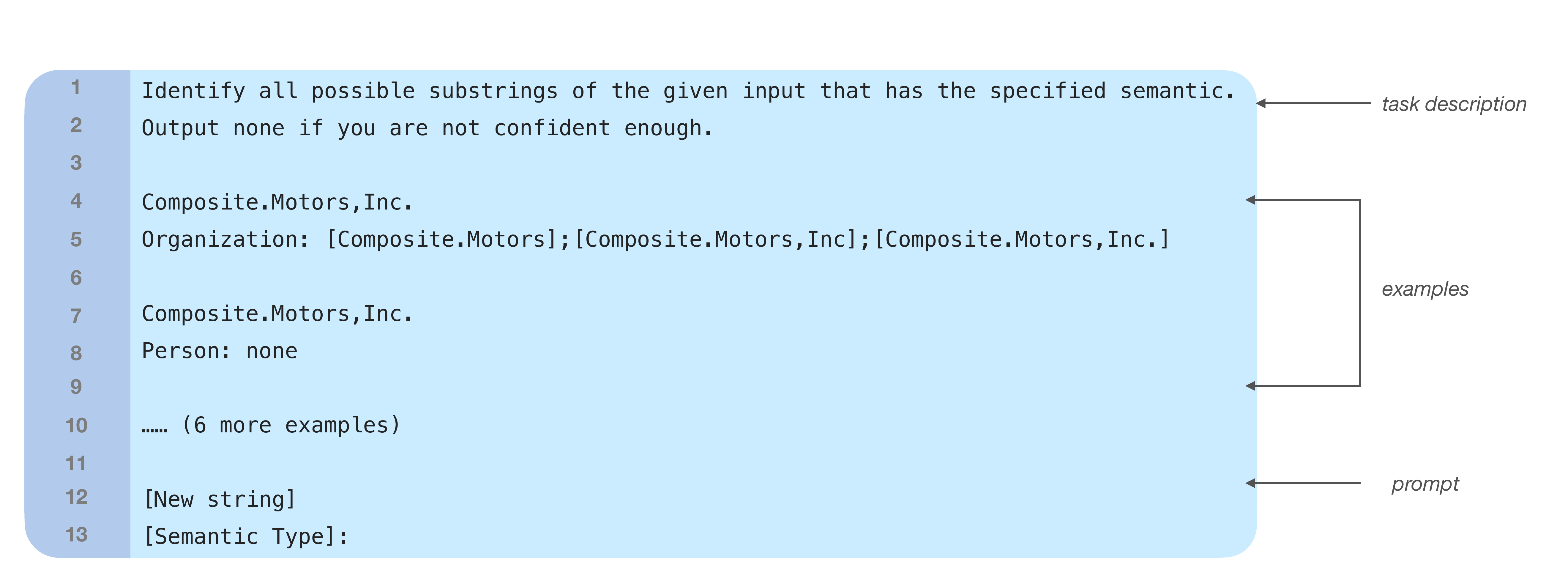}
    \caption{GPT-3 input structure for identifying substring of specific semantics. [New string] is a placeholder for the string we are querying about, and [Semantic Type] is the semantics we are asking the model to identify.}
    \label{fig:semantic_prompt}
    \vspace{-0.5cm}
\end{figure}

\paragraph{{\bf Implementation of checking observational equivalence}} In the \textsc{GetNextCompletion} procedure (Figure~\ref{fig:hole_synthesis}), we use the set $\ex_{\star}$ to prune out programs that are observationally equivalent to previously synthesized programs.  In Figure~\ref{fig:goal_synthesis}, $\ex_{\star}$ corresponds to all substrings of the negative examples $\ex^-$, but this set might contain too many strings in practice, leading to considerable overhead in the observational equivalence check. To address this issue, we only obtain the substring of the negative examples that are relevant to the specific hole under consideration. Specifically, we identify the relevant substrings of the negative examples using the overapproximation of the sketch. If a negative example can already be rejected by the overapproximation of the sketch, it is safe to conclude that any instantiation of the sketch can reject this negative example and therefore that this example is irrelavant. For those negative examples that can be matched by the overapproximation, we identify substrings that might be matched by each hole of the sketch and use those to check observational equivalence. This strategy provides the full benefits of checking observational equivalence but significantly reduces overhead in some cases. We illustrate this discussion through the following example:

\begin{example}
Consider a synthesis task with the following positive and negative examples:
\begin{center}
\footnotesize
\begin{tabular}{cc}
    \toprule
     {\bf Positive Examples} & {\bf Negative Examples}  \\
     \midrule
     14+15 &  1+18 \\ 
     15+17 & 2+6 \\
     16+13 & 7-12 \\
     \bottomrule
\end{tabular} 
\end{center}
Suppose that the generated sketch is $\{\hole: {\sf Integer}\}[+]\{\hole: {\sf Integer}\}$. Using the overapproximation $(.*)[+](.*)$, we can already reject the negative example ``7-12'', so that negative example is not relevant for selecting different instantiations of the sketch. To find the rest of the relevant strings, notice that the overapproximation decomposes the first negative example by sending ``1'' to the first hole and ``18'' to the second hole (as the negative example for each of the holes). Following the same procedure, we obtain ``1'' and ``2'' as the relevant substring for the first hole. Now, considering the two synthesized programs $\{v: {\sf Integer} \ | \ x > 4\}$ and $\{v: {\sf Integer} \ | \ x > 5\}$ for the first hole, we can safely conclude that these two programs are observationally equivalent with respect to the negative examples since both programs reject the same set of negative examples (specifically, example ``1+18'' and ``2+6'').
\end{example}

\paragraph{{\bf Ranking heuristic}} Because there are often multiple semantic regexes that are consistent with the provided examples, it is important to use a ranking heuristic to choose between possible solutions. To this end, our method prioritizes  sketches that maximize the number of type annotations, and it prefers decompositions  that minimize the number of holes that are assigned empty strings as positive examples. Finally, when choosing between multiple regexes for a given hole, our algorithm prefers those with smaller ASTs, first ranked by height and then by the number of nodes.

\paragraph{{\bf Hyperparameters}} The \toolname system has a hyperparameter that controls the maximum depth of the synthesized programs for each hole, which is set to 4 by default.  For GPT-3 hyperparameters, we set the temperature to 0 (corresponding to greedy inference) and maximum length to 256.\footnote{{We also define the suitable stop sequences for each prompt to ensure GPT-3 doesn't have to generate 256 tokens.}}
% \gd{may want to comment that this is way longer than necessary and the model stops earlier?}
\section{Evaluation}
In this section, we describe the results of our experimental evaluation, which is designed to answer the following research questions:

\begin{itemize}[leftmargin=*]
    \item {\bf RQ1.}  How does our proposed data extraction approach compare against existing approaches?
    \item {\bf RQ2.} How does our synthesis algorithm compare to relevant baselines?
    \item {\bf RQ3.} How important are the different components of our synthesis algorithm for successfully solving these benchmarks?
    \item {\bf RQ4.} Do semantic regexes help humans more effectively solve data extraction tasks compared to standard regexes?
    \end{itemize}

\begin{table}[t!]
    \centering
    \caption{Description of the sample tasks used in the evaluation.}\label{tab:benchmarks_main}
    \footnotesize
    \vspace{-0.5cm}
    \begin{tabular}{cc}
    \toprule
    {\bf Domain} & {\bf Task Description} \\
    \midrule
    \multirow{2}{*}[2pt]{Business} & 
    {Restaurants that are created before 2000 or after 2010}  \\
    & {Businesses located in California}  \\
    \midrule
    \multirow{2}{*}{Sales} & {Products with Intel CPU that have more than 8GB memory}  \\
    & {TVs of  size less than 50' or resolution less than 1080P}  \\
    \midrule
    \multirow{2}{*}{Retail} & {Website titles that start with product names and are followed by a url}  \\
    & {Product names that contain measurement information}  \\
    \midrule
    \multirow{2}{*}{Marketing} & {Software engineering jobs that have specified working locations}  \\
    & {Business names with at least 3 words}  \\
    \midrule
    \multirow{2}{*}{Account}  & {Email addresses that have  a country domain and where the  username ends with number}  \\
    & {Software versions with at least 10 minor updates and more than one patch}  \\
    \midrule
    \multirow{2}{*}{Stock} & {Company names with 3-letter abbreviation}  \\
    & {Company names with ticker symbols containing special characters}  \\
    \midrule
    \multirow{2}{*}{Science} & {Location description with format State; County; More details}  \\
    & {Locations that are less than  11 miles from a road}  \\
    \midrule
    \multirow{2}{*}{Server} & {Apache logs with file id >=151000 or in the format of a zip file with id <= 50}  \\
    & {Photo files with numbers in their name}  \\
    \midrule
    \multirow{2}{*}{Museum} & {Purchase made by using three different funds}  \\
    & {Artwork with two artists born in the 14th century}  \\
    \midrule
    \multirow{2}{*}{Exhibition} & {Dimension of item between 10 and 50 inches }  \\
    & {Item that is associated with at least three categories}  \\
    \bottomrule
    \end{tabular}
\end{table}

\paragraph{{\bf Benchmarks}} To answer these questions, we evaluate \toolname on 50 data extraction tasks involving 10 different datasets, which cover a wide range of domains like sales, science, and art. These datasets contain many different string formats and involve a large variety of entities. Out of 50 tasks, 34 of the tasks require at least one built-in semantic type and 33 of the tasks require at least one custom semantic type. We consider an average of 5 data extraction tasks for each dataset and manually label a subset of the strings in each dataset as positive or negative {for each task}. Specifically, we use 6 of the manually labeled examples for training and the rest for testing.  Table~\ref{tab:benchmarks_main} describes some example tasks for each domain. 

%task is labeled with 6 training examples (including 3 positive examples and 3 negative examples) and 20 testing examples (including 10 positive examples and 10 negative examples). For each task, we train the tools using the training examples and evaluate the performance on the testing examples in terms of precision, recall, and $F_1$ score. Table~\ref{tab:benchmarks_main} describes some example tasks in each domain included in our evaluation. 

\paragraph{{\bf Experimental Setup}} All of our experiments are conducted on a machine with an Apple M2 Max CPU and 32GB of physical memory, running the macOS 13.2.1 operating system. We run GPT-3 through the OpenAI API. For each task, we set the timeout to 60 seconds (excluding the time to query OpenAI). 

\subsection{Comparison with Other Automated Data Extraction Techniques}

There are several techniques that can be used to automate data extraction tasks. To answer our first research question, we compare \toolname against the following alternative data extraction approaches:

%To handle string-matching tasks, one can use syntactic regex as a symbolic approach, invoke GPT-3 as a neural approach, or adopt another neuro-symbolic language. One might question why it is important to propose a new language that is a mix of symbolic and neural constraints. To answer this research question, we compare our approach against the following baselines:

\begin{itemize}[leftmargin=*]
    \item \textsc{ChatGPT-Regex-Synth}~\cite{chatgpt}: One way to automate data extraction is to synthesize standard regexes from positive and negative examples.  To evaluate this approach, we use ChatGPT to synthesize  standard regexes. If the synthesized regex  rejects the positive examples or accepts the negative examples, we ask ChatGPT to synthesize a different regex for up to ten iterations.\footnote{We set the temperature to $0.7$ for sampling.}
    \item \textsc{ChatGPT-Exec}~\cite{chatgpt}: Another way to automate data extraction  is to directly use ChatGPT. To evaluate this approach, we provide ChatGPT with positive and negative examples and then query it about strings in the test set. Hence, this approach  does not require synthesizing a program; {instead, it invokes ChatGPT on every test example}. 
    \item \textsc{FlashGPT}~\cite{flashgpt}: Recent work has proposed an extension of FlashFill, called FlashGPT, that can query GPT-3 in addition to performing syntactic transformations and pattern matching. For our third baseline, we also compare against FlashGPT by giving it positive and negative examples and then using it to synthesize a program in their DSL. 
    
    %A tool that synthesizes Flashfill programs that contain semantic transformation constructs using GPT-3. The synthesized program runs on the Prose SDK. 
\end{itemize}

\begin{comment}
\paragraph{Baseline Setup} We set up each of the baselines to run on our tasks in the following way\footnote{A more detailed description of the setup is provided in the supplementary material.}:
\begin{itemize}[leftmargin=*]
    \item \textsc{ChatGPT-Regex-Synth}: We provide the same set of in-context examples as the set used in our sketch generator, but label the ground truth as regular expressions other than a sketch. We sample 10 programs from ChatGPT and return the first program that matches the training positive examples and does not match the training negative examples as the synthesized regex.
    \item \textsc{ChatGPT-Regex-Synth}: We provide the training positive examples and the negative examples as the in-context examples and query ChatGPT to return `Yes' if the new testing string should be matched, and `No' otherwise.
    \item \textsc{FlashGPT}: \textsc{FlashGPT} is a tool that is tailored to string transformation tasks. To adopt the tool in our setting, we modify the output of the program from True or False to the full string for positive examples and the empty string for negative examples.
\end{itemize}
\end{comment}

\begin{table}[t]
    \centering
    \footnotesize
    \caption{Evaluation results for \toolname and data extraction baselines. P means precision and R means recall.}\label{tab:eval1}
    \vspace{-0.5cm}
    \begin{tabular}{ccccccc}
    \toprule
    {\bf Tool} & {\bf \# Finished} & {\bf P } & {\bf R} & $\mathbf{F_1}$ & {\bf Synth Time (s)} & {\bf Matching Engine} \\
    \midrule
    \textsc{ChatGPT-Regex-Synth} & 23/50 & 0.60 & 0.40 & 0.44 & - & Regex\\
    \textsc{ChatGPT-Exec} & - & 0.60 & 0.77 & 0.65 & - & ChatGPT \\
    \textsc{FlashGPT} & 15/50 & 0.45 & 0.83 & 0.58 & 3.16 & FlashGPT DSL \\
    \midrule
    \toolname & 48/50 & 0.94 & 0.84 & 0.87 & 4.96 & Semantic Regex \\
    \bottomrule
    \end{tabular}
    \vspace{-0.5cm}
\end{table}

\paragraph{Main results} Our main results are summarized in Table~\ref{tab:eval1}. We evaluate each tool in terms of precision, recall, and F1 score on the test set as well as synthesis time and number of benchmarks solved. The {\bf P}, {\bf R}, and $\mathbf{F_1}$ columns represent the precision, recall, and $F_1$ score on the test set. \toolname achieves the highest precision, recall, and $F_1$ score among all the alternative data extraction approaches. In particular, \toolname outperforms the second best approach, namely \textsc{ChatGPT-Exec}, by 22\% in terms of $F_1$ score. While \textsc{ChatGPT-Exec} and {\sc FlashGPT} have fairly high recall, they have low precision. {\sc ChatGPT-Regex-Synth} has similar precision to {\sc ChatGPT-Exec} but has very low recall on the test set. Finally, {\sc FlashGPT} and \toolname are close in terms of recall, but \toolname significantly outperforms {\sc FlashGPT} in terms of precision (for benchmarks that both tools  can synthesize within the time limit).

Next, the column labeled ``\# Finished'' in Table~\ref{tab:eval1} shows the number of  tasks that each tool is able to solve. For \toolname and {\sc FlashGPT}, solving a benchmark means  they were able to find a program consistent with the positive and negative examples within the 60-second time limit. Solving a benchmark for {\sc ChatGPT-Regex-Synth} means finding a regex consistent with the examples within 10 iterations.\footnote{Recall we keep querying for a different regex for up to 10 times if the synthesized regex does not match the examples.} Since {\sc ChatGPT-Exec} does not perform synthesis, this column is not applicable to it. 
Among all the synthesis-based approaches, \toolname terminates for 48 out of 50 tasks, which is around twice as many as \textsc{ChatGPT-Regex-Synth} and around 3 times as many as \textsc{FlashGPT}. 

Finally, the column labeled ``Synth time'' shows the synthesis time in seconds for {\sc FlashGPT} and \toolname. Since we exclude the time to query OpenAI from synthesis time (this only takes at most a few seconds), this column is not applicable to {\sc ChatGPT-Regex-Synth}. As we can see from this column,  the synthesis time of \toolname is around 5 seconds, so it takes slightly longer than \textsc{FlashGPT} (which takes around 3 seconds) for the 14 tasks that both of the tools can solve. However, \toolname is able to synthesize a program for three times as many tasks as \textsc{FlashGPT}.

\paragraph{Failure Analysis for the baselines} To provide some insight into the shortcomings of existing approaches, we briefly discuss the failure cases of the  baselines.  As expected, \textsc{ChatGPT-Regex-Synth} struggles with tasks that are hard to represent as regular expressions, such as matching all businesses that are in California. Although \textsc{FlashGPT} combines neural and symbolic constructs, its neural component 
processes positive and negative examples rather than semantic types. In other words, the neural constructs directly query GPT with positive and negative examples rather than querying whether a  string matches a certain  type. As a result, it frequently generates trivial programs that directly invoke GPT with the training examples as input. Hence, it ultimately ends up sharing the same limitations as  \textsc{ChatGPT-Exec}.

\paragraph{Failure analysis for the \toolname} We examined instances where \toolname is unable  to complete the synthesis task within the allotted time and found that it encounters difficulties in tasks that demand a higher level of granularity from  semantic pattern matching. For example, consider a task that involves finding restaurant names containing a person's name. For the positive example ``Alice Chinese Bistro'', the entity matcher  may fail to recognize ``Alice'' as a person's name, causing \toolname to fail to synthesize a program consistent with all examples. 

\subsection{Comparison with Other Semantic Regex Synthesis Techniques}\label{sec:eval_comp}

\begin{table}[t]
    \centering
    \footnotesize
    \caption{Evaluation results for our tool and synthesis baselines. P means precision and R means recall. }\label{tab:eval2}
    \vspace{-0.5cm}
    \begin{tabular}{cccccc}
    \toprule
    {\bf Tool} & {\bf \# Finished} & {\bf P } & {\bf R} & $\mathbf{F_1}$ & {\bf Synth Time (s)} \\
    \midrule
    \textsc{ChatGPT-Synth}
    % \tablefootnote{\revise{We also run a variant that attempts to repair wrong programs generated by ChatGPT using the technique described in Section~\ref{sec:sketch_gen}. This variants can synthesize a semantic regex consistent with the examples for 12 of the 50 benchmarks. }} 
    & 6/50 & 0.76 & 0.67 & 0.71 & - \\
    \toolname-\textsc{NoSketch} & 12/50 & 0.79 & 0.85 & 0.79 & 15.27 \\
    \midrule
    \toolname & 48/50 & 0.94 & 0.84 & 0.87 & 4.96 \\
    \bottomrule
    \end{tabular}
    \vspace{-0.5cm}
\end{table}

To answer our second research question, we compare the neural-guided synthesis algorithm of \toolname against the following two purely-neural or purely-symbolic baselines:

\begin{itemize}[leftmargin=*]
    \item \textsc{ChatGPT-Synth}~\cite{chatgpt}:  To evaluate whether a purely neural synthesizer can solve these benchmarks, we use  ChatGPT to create a synthesizer for semantic regexes.  Specifically, our \textsc{ChatGPT-Synth} baseline queries ChatGPT to synthesize a \emph{semantic regex} that matches all positive examples and rejects all negative examples. If the generated semantic regex is inconsistent with the examples, we query it again for a different one. We repeat this process for up to 10 times, as done with  our {\sc ChatGPT-Regex-Synth} baseline in the previous subsection.
    \item \textsc{\toolname-NoSketch}: To evaluate a semantic regex synthesis without neural sketch generation, we create a variant of \toolname that does not start with a sketch (i.e., it uses $\{\hole: \textsf{Any}\}$ as the sketch). 
\end{itemize}

The results of this comparison are presented in Table~\ref{tab:eval2}. As we can see from the ``\# Finished'' column, {\sc ChatGPT-Synth} can synthesize a semantic regex consistent with the examples for only 6 of the 50 benchmarks within 10 iterations. On the other hand, {\sc \toolname-NoSketch} times out on the majority of benchmarks and only finds a consistent regex for 12 of the 50 benchmarks. Furthermore, for semantic regexes that both  {\sc \toolname-NoSketch}  and \toolname can synthesize, \toolname is significantly faster. 
Table~\ref{tab:eval2} also shows that \toolname  outperforms both of these synthesizers in terms of $F_1$-score when evaluated on the test data. In particular, among all tasks that can be solved by both \toolname and \textsc{ChatGPT-Synth}, \toolname achieves an $F_1$ score of $0.94$ versus $0.71$, and, among tasks that can be solved by both \toolname and \textsc{\toolname-NoSketch}, \toolname achieves an $F_1$ score of $0.88$ versus $0.84$.

\subsection{Ablation Study}

% In this section, we describe two ablation studies to evaluate the relative impact of the different components of our proposed synthesis technique. Specifically, we consider the following ablations:

In this section, we describe two ablation studies to assess the relative impact of different components of \toolname: one evaluates the impact of the synthesis techniques proposed in Section~\ref{sec:top-level}-\ref{sec:compsitional}, and the other one evaluates the impact of generating sketches rather than concrete regexes.

\paragraph{Ablations of components of the synthesis techniques} To evaluate the effectiveness of the proposed synthesis techniques, we consider the following ablations:

\begin{itemize}[leftmargin=*]
    \item \textsc{\toolname-NoDecomp}: A variant of \toolname that does not  perform compositional sketch completion. In particular, this variant does not infer positive examples for each hole. 
 \item \textsc{\toolname-NoTypedHole}: A variant of \toolname that does not use typed sketches. That is, each hole in the sketch is annotated with type \textsf{Any}. 
    \item \textsc{\toolname-NoLocateError}: A variant of \toolname that does not perform error localization for sketch repair. Instead, it queries GPT-3 for a new sketch through sampling. 

    \item \textsc{\toolname-NoTypeSystem}: A variant of \toolname that does not perform type-directed synthesis. 
\end{itemize}

%Figure~\ref{fig:eval3} plots the number of solved tasks against cumulative running time for \toolname and its variants. In this context, a sketch is considered solved if the synthesis engine can find an instantiation of the sketch that is consistent with the training examples. 

\begin{figure}
    \centering
    \vspace{-0.5cm}
    \hspace{-1.5cm}
    \begin{minipage}[t]{0.5\textwidth}
    \definecolor{darkblue}{rgb}{0,0.24706,0.36078}
\definecolor{darkpurple}{rgb}{0.3451,0.31373,0.55294}
\definecolor{darkpink}{rgb}{0.73725,0.31373,0.56471}
\definecolor{orangered}{rgb}{1,0.38824,0.38039}
\definecolor{yellow}{rgb}{1,0.65098,0}

\begin{tikzpicture}[scale=0.6]
\begin{axis}[
    ymax=500,
    y=0.01cm,
    x=0.18cm,
    legend cell align = left,
    legend pos = outer north east,
    legend style = {
        nodes={scale=0.8, transform shape},
        at={(0.22,0.98)},
        legend columns=1,
        anchor=north,
    },
    xlabel style={yshift=1mm},
    ylabel = Cumulative Time (s),
    xlabel = \# Completed Benchmarks,
    xmax = 51,
    xmin = -4
]
    \legend{{\sc \toolname},{\sc \toolname-NoDecomp},{\sc \toolname-NoLocateError},{\sc \toolname-NoTypedHole},{\sc \toolname-NoTypeSystem}}
 \addplot[smooth, line width=0.4mm, mark=square, mark options={fill=yellow}, mark size=0.8pt, yellow] coordinates {
(1, 0.03965687752) 
(2, 0.1004858017) 
(3, 0.2213199139) 
(4, 0.3682119847) 
(5, 0.5169529916) 
(6, 0.6668088437) 
(7, 0.8301157952) 
(8, 1.017044544) 
(9, 1.256562472) 
(10, 1.58267951) 
(11, 1.921844483) 
(12, 2.270258427) 
(13, 2.6873312) 
(14, 3.160350084) 
(15, 3.729110241) 
(16, 4.350345373) 
(17, 4.974715233) 
(18, 5.631301165) 
(19, 6.398749352) 
(20, 7.247143269) 
(21, 8.133033991) 
(22, 9.064376831) 
(23, 10.07906079) 
(24, 11.20059371) 
(25, 12.45900154) 
(26, 13.73647952) 
(27, 15.16658664) 
(28, 17.2307446) 
(29, 19.30641842) 
(30, 21.5041275) 
(31, 23.73273444) 
(32, 26.3826716) 
(33, 29.37882757) 
(34, 32.64707065) 
(35, 36.1312778) 
(36, 40.68897176) 
(37, 45.93495774) 
(38, 53.03780294) 
(39, 60.8009572) 
(40, 69.01777721) 
(41, 77.88430429) 
(42, 86.90586329) 
(43, 99.96785427) 
(44, 113.3079002) 
(45, 127.5502944) 
(46, 146.8048055) 
(47, 190.0116797) 
(48, 238.0030785) 
};
\addplot[smooth, line width=0.4mm, mark=triangle, mark options={fill=darkblue}, mark size=0.8pt, darkblue] coordinates {
(1, 0.1447861195) 
(2, 0.5516242981) 
(3, 2.512999296) 
(4, 4.757796287) 
(5, 7.387031555) 
(6, 25.89839745) 
(7, 49.61691642) 
(8, 76.09498739) 
(9, 126.4115302) 
(10, 181.0694923) 
};
\addplot[smooth, line width=0.4mm, mark=x, mark options={fill=darkpink}, mark size=0.8pt, darkpink] coordinates {
(1, 0.04427695274) 
(2, 0.1190478802) 
(3, 0.2663710118) 
(4, 0.4342792035) 
(5, 0.6286289693) 
(6, 0.8342490197) 
(7, 1.085031986) 
(8, 1.379728079) 
(9, 2.060221195) 
(10, 2.8637712) 
(11, 3.836070061) 
(12, 4.810188055) 
(13, 5.864868164) 
(14, 7.895837307) 
(15, 10.11360312) 
(16, 12.67210507) 
(17, 15.51138806) 
(18, 19.08019018) 
(19, 22.89815497) 
(20, 27.72394013) 
(21, 33.01735807) 
(22, 39.22887921) 
(23, 48.10181832) 
(24, 57.84739232) 
(25, 97.79206943) 
(26, 148.2325003) 
(27, 200.0532043) 
(28, 252.257539) 
};
\addplot[smooth, line width=0.4mm, mark=diamond, mark options={fill=orangered}, mark size=0.8pt, orangered] coordinates {
(1, 1.050572157) 
(2, 2.43678689) 
(3, 4.004702806) 
(4, 5.968386888) 
(5, 7.977417945) 
(6, 10.04178882) 
(7, 12.80834174) 
(8, 16.07967591) 
(9, 19.37424803) 
(10, 22.89874911) 
(11, 27.4166522) 
(12, 32.26391625) 
(13, 37.17088222) 
(14, 43.10218406) 
(15, 49.10283208) 
(16, 55.31619596) 
(17, 61.61715293) 
(18, 68.25550985) 
(19, 75.18367195) 
(20, 82.74151587) 
(21, 91.74118376) 
(22, 100.9624867) 
(23, 111.8156557) 
(24, 123.0637774) 
(25, 134.5358362) 
(26, 153.0063622) 
(27, 176.3082881) 
(28, 200.2615402) 
(29, 230.5961502) 
(30, 265.541584) 
(31, 307.1880932) 
(32, 357.961987) 
(33, 412.6069279) 
};
\addplot[smooth, line width=0.4mm, mark=pentagon, mark options={fill=darkpurple}, mark size=0.8pt, darkpurple] coordinates {
(1, 0.06372499466) 
(2, 0.1382350922) 
(3, 0.217040062) 
(4, 0.3390021325) 
(5, 0.4766449929) 
(6, 0.6230509282) 
(7, 0.7910897733) 
(8, 1.020676613) 
(9, 1.335031509) 
(10, 1.786788702) 
(11, 2.360619545) 
(12, 3.009411573) 
(13, 3.679698467) 
(14, 4.70019269) 
(15, 6.001746893) 
(16, 17.74664283) 
(17, 34.43334985) 
(18, 54.73396993) 
(19, 82.25398994) 
(20, 136.1226869) 
(21, 191.1110468) 
};
\end{axis}\end{tikzpicture}
    \vspace{-0.8cm}
    \caption{Solved tasks over time.}
    \label{fig:eval3}
    % \vspace{-0.5cm}
    \end{minipage}
    \hspace*{0.3cm}
    \begin{minipage}[t]{0.43\textwidth}
    \centering
    \vspace*{-3.4cm}
    \scriptsize
    \begin{tabular}{cccc}
    \toprule
    Task & Manual-Regex & Manual-SemRegex & \toolname \\ 
    \midrule 
     1 & 0.31 & 0.93 & 1.00 \\
     2 & 0.65 & 0.65 & 0.89 \\
     3 & 0.55 & 0.81 & 0.89 \\
     4 & 0.63 & 0.71 & 0.89 \\
     \midrule
     Average & 0.54 & 0.78 & 0.92 \\
     \bottomrule
    \end{tabular}
    \vspace{-0.1in}
    \caption{Average $F_1$ scores achieved by manually-written regexes,  semantic regexes, and synthesized semantic regexes.}
    \label{tab:user_study_res}
    \end{minipage}
    \vspace{-0.6cm}
\end{figure}

The results of this ablation study are presented in Figure~\ref{fig:eval3}, which shows the number of benchmarks completed (x-axis) within the given time limit (y-axis). As we can see from the gap between the five lines, \toolname is significantly faster than all other variants and achieves a speedup of 14$\times$  compared to the second-fastest baseline, \textsc{\toolname-NoTypedHole}. Hence, this ablation study shows that all algorithmic components proposed in this paper are important for speeding up the synthesis.

\paragraph{Ablations of sketch generations.} To understand the significance of generating sketches as opposed to concrete semantic regexes, we introduce a new baseline named \textsc{ChatGPT-Synth-Repair}. This baseline extends the \textsc{ChatGPT-Synth} baseline from Section~\ref{sec:eval_comp} with program repair. Specifically, it first generates a concrete program using ChatGPT (using a similar prompt as the \textsc{ChatGPT-Synth} baseline). If the generated program does not satisfy all the positive and negative examples provided, it then performs the error localization and repair strategies presented in Section~\ref{sec:sketch_gen}.

\begin{table}[t]
    \centering
    \footnotesize
    \caption{Evaluation results for our tool and the no-sketch variant. P means precision and R means recall. }\label{tab:eval3}
    \vspace{-0.5cm}
    \begin{tabular}{cccccc}
    \toprule
    {\bf Tool} & {\bf \# Finished} & {\bf P } & {\bf R} & $\mathbf{F_1}$ & {\bf Synth Time (s)} \\
    \midrule
    \textsc{ChatGPT-Synth}
    & 6/50 & 0.76 & 0.67 & 0.71 & - \\
    \textsc{ChatGPT-Synth-Repair} & 12/50 & 0.76 & 0.67 & 0.71 & - \\
    \midrule
    \toolname & 48/50 & 0.94 & 0.84 & 0.87 & 4.96 \\
    \bottomrule
    \end{tabular}
    % \vspace{-0.5cm}
\end{table}

The results of this ablation study are presented in Table~\ref{tab:eval3}. For clarity, we also include the results of \textsc{ChatGPT-Synth} and \toolname from Section~\ref{sec:eval_comp} to show the difference evaluation results. This ablation leads to the two following observations: 
\begin{itemize}[leftmargin=*]
    \item \textsc{ChatGPT-Synth-Repair} solves 6 more benchmarks compared to \textsc{ChatGPT-Synth}. This shows that our sketch repair technique can also be generalized to concrete program repair. 
    \item Although \textsc{ChatGPT-Synth-Repair} exhibits superior performance over \textsc{ChatGPT-Synth}, it is not comparable to the performance of \toolname, which leverages sketches for synthesis. This underscores the pivotal role sketches play in enhancing the tool's efficacy. Upon analysis, we find that \textsc{ChatGPT-Synth-Repair} is able to accurately locate the error when it does not generate the desired program. However, ChatGPT struggles to generate a new program that precisely separates positive from negative examples. In contrast, with \toolname, since we produce sketches, ChatGPT only needs to generate segments of the program it is confident about, delegating the uncertain parts or those demanding intricate reasoning to the program synthesizer.
\end{itemize}

\subsection{User Study}\label{sec:user_study}
We conducted a user study to assess the efficacy of semantic regexes in aiding humans with data extraction tasks compared to standard regexes. We recruited 13 participants, consisting of 3 CS undergraduate students, 6 CS graduate students, and 4 professional software engineers who regularly use regexes in their work. We asked each participant to complete 4 data extraction tasks by writing a regex. % that could successfully match a set of positive strings and reject a set of negative strings. 
%For each task, the participants were presented with 3 sample positive strings and 3 sample negative strings 
%and informed that their programs would be evaluated on an unseen test to measure their performance. 
The participants were given 5 minutes for each task and asked to write standard regexes for two randomly chosen tasks (out of the 4 total tasks)  and semantic regexes for the other two. The four tasks used in the study are simplified versions of the benchmarks used in our evaluation --- we intentionally simplified the tasks so that they are doable within 5 minutes.

%selected from the benchmarks  used (we simplified the tasks so that it is possible for users to finish within 5 minutes). 

\paragraph{Setup} To conduct this user study, we developed a command-line interface for \toolname. For each task, the interface initially displays the prompt for the task (including 3 positive and negative examples) and then asks the user to input their answer. The tool randomly determines whether the answer should be a standard or semantic regex and only accepts user answers in the correct format. 
Upon entering a regex, the interface evaluates it against the test set and informs the user of their regex's performance, allowing unlimited attempts to enter a new regex within the 5-minute time limit. The details of the user study protocol are provided in the supplementary material. 

%in terms of the number of positive examples accepted and the number of negative examples rejected. Within the 5-minute time limit, the users are granted an unlimited number of attempts if they have not yet achieved perfect performance on both the sample set and the unseen test set.

\paragraph{{\bf Results}} We evaluate the quality of the regexes in terms of their $F_1$ score on the test set. For each task, Table~\ref{tab:user_study_res} presents $F_1$ scores for (a) manually-written standard  regexes  (``Manual-Regex''), (b) manually-written \emph{semantic} regexes (``Manual-SemRegex''), and (c) semantic regexes generated automatically by \toolname (the ``\toolname'' column). Since some of the manually-written regexes have a precision or recall score of $0$, the $F_1$ score is undefined. In Table~\ref{tab:user_study_res}, we only show average $F_1$ score across regexes for which the $F_1$ score is defined.

As we can see from Figure~\ref{tab:user_study_res}, manually-written \emph{semantic} regexes achieve a better overall $F_1$ score (0.78) compared to standard regexes, for which the $F_1$ score is $0.54$. We ran a two-way ANOVA to find the most significant factor affecting the $F_1$ score. In particular, we model the $F_1$ score as the dependent variable and the type of tool and task as independent variables. The ANOVA analysis shows that the ``task'' variable has a high p-value of 0.57, which indicates it does not have a significant impact on the $F_1$ score. On the other hand, the ``type of tool'' variable has a low p-value of 0.003, suggesting that the type of tool used has a significant impact on user performance. The analysis result indicates that participants are more effective at performing these types of data extraction tasks using semantic regexes than with standard regexes. Another interesting aspect of Figure~\ref{tab:user_study_res}  is that the semantic regexes learned by \toolname seem to be \emph{even} more effective than manually-written semantic regexes. In particular, for these four tasks, \toolname learns regexes that achieve an overall $F_1$ score of $0.92$ compared to the  $F_1$ score ($0.78$) of manually-written semantic regexes. This result suggests that our proposed learning technique has the potential to improve productivity even for expert users who are generally comfortable with writing regexes. 

%at data extraction tasks when using semantic regexes, at least for the four representative tasks from our benchmark set. \jcedit{Furthermore, our study results show that compared to the manually-written semantic regexes, \toolname was able to synthesize semantic regexes that had higher performance across all 4 tasks. This small user study also suggests that our synthesis algorithm is valuable, enabling users (even those familiar with regex) to achieve better results for data extraction tasks. 
%Out of the 28 standard regexes written by the users, only 18 could match at least one positive example or reject one negative example. In contrast, when using semantic regexes, 23 regexes were capable of matching at least one positive example or rejecting one negative example. For programs yielding meaningful results on the test set, the standard regex programs achieved an average $F_1$ score of $0.55$ across tasks, whereas the semantic regex programs reached an average $F_1$ score of $0.75$. This user study demonstrates that (1) semantic regexes can assist users in writing better generalizable programs on data extraction tasks, and (2) our synthesis algorithm is valuable, enabling users (even those familiar with regex) to achieve better results in data extraction. A more detailed presentation of our user study can be found in the supplementary material. 
\section{Related Work}

In this section, we survey related work on program synthesis and data extraction. 

\paragraph{{\bf Learning regexes from examples}} There is a large body of prior research on learning regular expressions from positive and negative examples~\cite{alquezar94, firoiu98, angluin1987, gold1978, rivest1989, DFA1, DFA2}. Our work builds on existing works that prune partial programs by evaluating the examples with respect to over- and under-approximations~\cite{alpharegex, regel, opsynth}. In this work, we not only use the over-approximations for pruning but also for decomposing the synthesis tasks.

\paragraph{{\bf Information Extraction from Semi-Structured Data}} Past work has investigated similar extraction tasks, particularly for extracting lists from web sources~\cite{pasupat-liang-2014-zero, webqa, hyb, freedom}, answering questions based on tables~\cite{pasupat-liang-2015-compositional}, and general information extraction from tabular data~\cite{fonduer, le14}. Recent work has specifically employed LLMs to extract information from tables~\cite{binder} or raw text \cite{dunn2022structured}. Despite the prevalence of neural-based approaches that emphasize data semantics, our work uniquely targets the integration of both semantic and symbolic aspects of the data structure.

\paragraph{{\bf Neurosymbolic DSLs}} Recent work has considered so-called \emph{neurosymbolic DSLs} with both standard language constructs and neural components~\cite{nmn, nmn2, neuralregex, houdini, ntpt, near, webqa, Bastani2022, referexpression, flashgpt, binder}. Among these, most relevant to our approach are FlashGPT~\cite{flashgpt} and Binder~\cite{binder}. FlashGPT augments the DSL used in Flashfill~\cite{flashfill} with  semantic transformation operators that can be used to reason about  the semantic properties of the input. However, FlashGPT relies on in-context examples and does not utilize explicit semantic types, which hinders its ability to reason about combined semantic and symbolic properties. On the other hand, \textsc{Binder}~\cite{binder} proposes a new program structure that extends programming languages, such as SQL, with a function that allows querying large language models (in particular, Codex). However, the constructs proposed in \textsc{Binder} focus mainly on SQL-related tasks and do not transfer well to the string-matching domain. 

\paragraph{{\bf Program Synthesis Using LLMs}} The growing interest in leveraging LLMs for program synthesis~\cite{binder, codegen, programsynthesisllm, codexpaper,docprompting} stems from general-purpose models like ChatGPT and Codex demonstrating code generation capabilities from various specifications, including natural language and input-output examples. However, these models often generate code that violates syntactic and semantic rules due to their limited understanding of program syntax and semantics. To address this, several approaches~\cite{jigsaw, marriage, poesia2022synchromesh} integrate LLMs with symbolic methods like program analysis  to improve code quality. In our work, we use LLMs to generate sketches and introduce a sketch repair technique to handle cases where the LLM fails to generate accurate sketches.

\paragraph{{\bf Compositional program synthesis}} Various approaches have been proposed for compositional program synthesis~\cite{bansal22,huang20,lambda2,flashmeta}. Among these works, both $\lambda^2$~\cite{lambda2} and FlashMeta~\cite{flashmeta} perform compositional PBE by inferring input-output examples for sub-programs using the inverse semantics. In another example, Raza et al.~\cite{cps} rely on the natural language description to decompose the synthesis problems into smaller sub-problems. Furthermore, Zhang et al.~\cite{clis} decompose the synthesis task into simpler sub-problems in the domain of UDF-to-SQL translation using a dataflow graph. Our work differs from prior research by presenting a new decomposition strategy on a typed sketch in the context of synthesizing string-matching programs. While our decomposition approach helps reject incorrect programs using inferred positive examples, the full result must still be tested against the negative examples to ensure correctness.

\paragraph{{\bf Semantic Checks for String Matching}} There has been prior work in combining string matching with semantic matching~\cite{kat, katmodulo}; for example, Kleene algebra with tests (KATs)~\cite{kat} combines Kleene and Boolean algebra. While our semantic matching construct can be conceptually viewed as a semantic guard for string matching, one key difference is that the predicate (i.e. the ``test'') part of the language in KATs is restricted to boolean algebra, whereas our vocabulary of predicates is much richer, including function invocations and machine learning models. Furthermore, the intended application domains are quite different: our proposed semantic regexes are intended for textual data extraction, whereas KATs have traditionally been used in the context of verification.
\section{Conclusion}
We have presented \toolname, a new synthesis-powered system for data extraction. The key idea behind \toolname is the concept of \emph{semantic regexes}, which augments the syntactic pattern matching capabilities of regexes with a semantic pattern matching construct of the form $\{v: \tau \ | \ \phi \}$ which matches strings that have entity type $\tau$ and that satisfy logical predicate $\phi$ when interpreted as an instance of $\tau$. As shown in our user study from Section~\ref{sec:user_study}, semantic regexes allow users to more easily perform data extraction tasks that are hard to do using standard regular expressions.

In addition to proposing semantic regexes, we have also described a learning algorithm that can synthesize semantic regexes from examples. Our synthesis algorithm is neural-guided and uses a LLM to generate a \emph{typed sketch} where unknown parts of the regex have useful type annotations that are used to guide the search. Our synthesis algorithm is compositional and  uses type-directed reasoning to find a completion of each hole in the sketch. 
Our evaluation shows that our proposed approach outperforms alternative data extraction techniques in terms of precision, recall, and $F_1$ score. Our evaluation also shows the advantages of combining neural-guided sketch generation with type-directed compositional synthesis in terms of synthesis time.

%% Acknowledgments
\begin{acks}                            %% acks environment is optional
                                        %% contents suppressed with 'anonymous'
  %% Commands \grantsponsor{<sponsorID>}{<name>}{<url>} and
  %% \grantnum[<url>]{<sponsorID>}{<number>} should be used to
  %% acknowledge financial support and will be used by metadata
  %% extraction tools.
  This material is based upon work supported by the
  \grantsponsor{GS100000001}{National Science
    Foundation}{http://dx.doi.org/10.13039/100000001} under Grant
  No.~\grantnum{GS100000001}{nnnnnnn} and Grant
  No.~\grantnum{GS100000001}{mmmmmmm}.  Any opinions, findings, and
  conclusions or recommendations expressed in this material are those
  of the author and do not necessarily reflect the views of the
  National Science Foundation.
\end{acks}

%% Bibliography
\bibliography{main.bib}

% Appendix
\pagebreak
\appendix
\section{Proofs}

\begin{lemma}\label{lemma:type_semantic_correspondance_1}
Let $\regex$ be a semantic regex of type $\type$ and $s$ be a arbitrary string such that ${\tt SemanticType}(s) \neq {\sf CharSeq}$, if $\sem{r}(s)$ evaluates to ${\tt True}$, then ${\tt Semantic}({\tt SemanticType}(s)) \subtype \type$.
\end{lemma}

\begin{proof}
We prove this lemma by doing structural induction on $\regex$. 

\paragraph{{\bf Base Case 1:} $\regex = c$, where $c$ is a constant} Given the type of $\regex$ is derived using the typing rule \textsc{Const-Semantic}, we have $\tau = {\tt Semantic}({\tt SemanticType}(c))$. Furthermore, since $c \equiv s$, we can conclude ${\tt Semantic}({\tt SemanticType}(s)) \subtype \type$.

\paragraph{{\bf Base Case 2:} $\regex = cc$, where $cc$ is a character class} Here, we only focus on the case where $\regex$ is $<{\tt Num}>$. Following the typing rule \textsc{CC-Num}, we know $\type = {\tt Semantic}({\sf Number})$. Since $<{\tt Num}>$ matches a number of length 1, we know $s$ has the semantic type ${\sf Number}$. Therefore, we have ${\tt Semantic}({\tt SemanticType}(s)) \subtype \type$.

\paragraph{{\bf Base Case 3:} $\regex = \matchsemq{\type}{}$} Using the typing rules \textsc{MatchSem}, we know $\regex$ has the type ${\tt Semantic}(\type)$. According to the semantics presented in Figure~\ref{fig:dsl_semantics_match}, $\matchsemq{\type}{}$ matches all strings $s'$ of semantic type $\type$. Since $\matchsemq{\type}{}$ matches $s$, ${\tt Semantic}({\tt SemanticType}(s)) \subtype \type$. 

\paragraph{{\bf Inductive hypothesis:}} Assuming this lemma holds for all regexes $\regex_1, \cdots, \regex_n$.

\paragraph{{\bf Inductive case:}} We show that all programs constructed using the programs in the inductive hypothesis also satisfy this lemma by considering all the possible top-level constructs in the grammar.

\begin{itemize}
    \item $\regex = \neg\regex_1$. Using the typing rule \textsc{Not}, we know $\regex$ has the type ${\sf Any}$. Since any string has a semantic type $\type_s$ such that $\type_s \subtype {\sf Any}$,  we conclude ${\tt Semantic}({\tt SemanticType}(s)) \subtype \type$. 
    \item $\regex = \regex_1*$. Using the typing rule \textsc{Star-2} in Figure~\ref{fig:typing_rules_more}, we derive the type of$\regex$ to be ${\sf Any}$. Since any string has type $\type_s$ such that $\type_s \subtype {\sf Any}$,  we conclude ${\tt Semantic}({\tt SemanticType}(s)) \subtype \type$. 
    \item $\regex = \regex_1?$. Using the typing rule \textsc{Optional}, we know $\regex$ has the type ${\tt Optional}(\tau_1)$, where $\tau_1$ is the type of $\regex_1$. Using the inductive hypothesis, we know that ${\tt Semantic}({\tt SemanticType}(s)) \subtype \type_1$. Since $\regex$ matches all strings that can be matched by $\regex_1$ plus the empty string, we conclude ${\tt Semantic}({\tt SemanticType}(s)) \subtype \tau$. 
    \item $\regex = \regex_1 \bigstar$, where $\bigstar \in \{+, \{k\}, \{k_1, k_2\}\}$. The proof is similar to that of $\regex_1*$.
    \item $\regex = \regex_1 \cup \regex_2$. From the typing rule \textsc{Or}, we know $\regex$ has the type $\type_1 \vee \type_2$, where $\type_1$ is the type of $\regex_1$ and $\type_2$ is the type of $\regex_2$. Using the semantics of the Or operator, we know that $s$ can either be matched by $\regex_1$ or $\regex_2$. If $s$ is matched by $\regex_1$, then using the inductive hypothesis, we know ${\tt Semantic}({\tt SemanticType}(s)) \subtype \type_1$; if $s$ is matched by $\regex_2$, then using the inductive hypothesis, we know ${\tt Semantic}({\tt SemanticType}(s)) \subtype \type_2$. Following the definition of type union, we can conclude that ${\tt Semantic}({\tt SemanticType}(s)) \subtype (\type_1 \vee \type_2)$. 
    \item $\regex = \regex_1 \cap \regex_2$. From the typing rule \textsc{And}, we know $\regex$ has the type $\type_1 \wedge \type_2$, where $\type_1$ is the type of $\regex_1$ and $\type_2$ is the type of $\regex_2$. Using the semantics of the And operator, we know that $s$ can be matched by both $\regex_1$ and $\regex_2$. Using the inductive hypothesis, we then know ${\tt Semantic}({\tt SemanticType}(s)) \subtype \type_1$ and ${\tt Semantic}({\tt SemanticType}(s)) \subtype \type_2$. Following the definition of type intersection, we can conclude that ${\tt Semantic}({\tt SemanticType}(s)) \subtype (\type_1 \wedge \type_2)$. 
    \item $\regex = \regex_1 \cdot \regex_2$. From the typing rule \textsc{Concat}, we know $\regex$ has the type ${\sf Any}$. Since any string has a semantic type $\type_s$ such that $\type_s \subtype {\sf Any}$, we conclude that ${\tt Semantic}({\tt SemanticType}(s)) \subtype \type$.
\end{itemize}
\end{proof}

\begin{lemma}\label{lemma:type_semantic_correspondance_2}
    Let $\regex$ be a semantic regex of $\tau$ and $s$ be an arbitrary string with no semantic meaning (i.e. ${\tt SemanticType}(s) = {\sf CharSeq}$), if $\sem{r}(s)$ evaluates to ${\tt True}$, then ${\tt SemanticType}(s) \subtype \type$.
\end{lemma}

\begin{proof}
    We prove this lemma by doing structural induction on $\regex$. In this proof, we only prove those cases where the program might be constructed to have a type $\type$ such that ${\sf CharSeq} \subtype \type$. 

    \paragraph{{\bf Base Case 1:} $\regex = c$, where $c$ is a constant} Since the type of $\regex$ is derived using the typing rule \textsc{Const-CharSeq}, we have $\type = {\sf CharSeq}$. Therefore ${\tt SemanticType}(s) \subtype \type$.
    \paragraph{{\bf Base Case 2:} $\regex = cc$, where $cc$ is a character class} Here, we only focus on the case where $\regex$ is not $<{\tt Num}>$. Following the typing rule \textsc{CC}, we know $\type = {\sf CharSeq}$. Since the character class does not have semantic meaning, we have ${\tt Semantic}({\tt SemanticType}(s)) \subtype \type$.
    \paragraph{{\bf Inductive hypothesis:}} Assuming this lemma holds for all regexes $\regex_1, \cdots, \regex_n$.
    \paragraph{{\bf Inductive case: }} We show that all programs constructed using the programs in the inductive hypothesis also satisfy this lemma by considering all possible top-level constructs in the grammar.

    \begin{itemize}
        \item $\regex = \neg\regex_1$. Using the typing rule \textsc{Not}, we know $\regex$ has the type ${\sf Any}$. Since any string $s$ has a type $\type_s$ such that $\type_s \subtype {\sf Any}$,  we conclude $\type_s \subtype \type$. 
        \item $\regex = \regex_1*$. Using the typing rule \textsc{Star-1} in Figure~\ref{fig:typing_rules_more}, we derive the type of$\regex$ to be ${\sf CharSeq}$. Since $s$ has the type ${\sf CharSeq}$, ${\tt SemanticType}(s) \subtype \type$. 
        \item $\regex = \regex_1?$. Using the typing rule \textsc{Optional}, we know $\regex$ has the type ${\tt Optional}(\tau_1)$, where $\tau_1$ is the type of $\regex_1$. Using the inductive hypothesis, we know that ${\tt SemanticType}(s) \subtype \type_1$. Since $\regex$ matches all strings that can be matched by $\regex_1$ plus the empty string, we conclude ${\tt SemanticType}(s) \subtype \tau$. 
        \item $\regex = \regex_1 \bigstar$, where $\bigstar \in \{+, \{k\}, \{k_1, k_2\}\}$. The proof is similar to that of $\regex_1*$.
        \item $\regex = \regex_1 \cup \regex_2$. From the typing rule \textsc{Or}, we know $\regex$ has the type $\type_1 \vee \type_2$, where $\type_1$ is the type of $\regex_1$ and $\type_2$ is the type of $\regex_2$. Using the semantics of the Or operator, we know that $s$ can either be matched by $\regex_1$ or $\regex_2$. If $s$ is matched by $\regex_1$, then using the inductive hypothesis, we know ${\tt SemanticType}(s) \subtype \type_1$; if $s$ is matched by $\regex_2$, then using the inductive hypothesis, we know ${\tt SemanticType}(s) \subtype \type_2$. Following the definition of type union, we can conclude that ${\tt SemanticType}(s) \subtype (\type_1 \vee \type_2)$. 
        \item $\regex = \regex_1 \cap \regex_2$. From the typing rule \textsc{And}, we know $\regex$ has the type $\type_1 \wedge \type_2$, where $\type_1$ is the type of $\regex_1$ and $\type_2$ is the type of $\regex_2$. Using the semantics of the And operator, we know that $s$ can be matched by both $\regex_1$ and $\regex_2$. Using the inductive hypothesis, we then know ${\tt SemanticType}(s) \subtype \type_1$ and ${\tt SemanticType}(s) \subtype \type_2$. Following the definition of type intersection, we can conclude that ${\tt SemanticType}(s) \subtype (\type_1 \wedge \type_2)$. 
        \item $\regex = \regex_1 \cdot \regex_2$. From the typing rule \textsc{Concat}, we know $\regex$ has the type ${\sf Any}$. Since any string has a semantic type $\type_s$ such that $\type_s \subtype {\sf Any}$, we conclude that ${\tt SemanticType}(s) \subtype \type$.
    \end{itemize}
\end{proof}

\begin{lemma}\label{lemma:type_semantic_correspondence}
Let $\regex$ be a semantic regex of type $\type$ and $s$ be a arbitrary string, if $\sem{r}(s)$ evaluates to ${\tt True}$, then ${\tt Semantic}({\tt SemanticType}(s)) \subtype \type$.
\end{lemma}

\begin{proof}
    We prove this lemma by dividing the $\type$ into two cases: the case where $\type$ is a semantic type (i.e. either ${\tt Semantic}(\type_s)$ or ${\tt Optional}({\tt Semantic}(\type_s))$), and the case where $\type$ is either ${\sf CharSeq}$ or ${\tt Optional}({\sf CharSeq})$. The first case is proved in Lemma~\ref{lemma:type_semantic_correspondance_1}, the second case is proved in Lemma~\ref{lemma:type_semantic_correspondance_2}, and therefore proved this theorem. 
\end{proof}

\begin{theorem}
Consider the synthesis problem with positive examples $\mathcal{E}^+$. Let $\sketch$ be a candidate sketch and let $r$ be a completion of $\sketch$ mapping each hole $ h_i$ in $\sketch$ to a semantic regex $r_i$. If $r$ satisfies all positive examples $\mathcal{E}^+$, then there exists some $\goal \in {\textsc{GetNextDecomp}}(\sketch, \mathcal{E}^+)$ such that every $r_i$ satisfies $\goal[h_i]$. 
\end{theorem}

\begin{proof}
We prove this theorem by doing structural induction on $\sketch$.

\paragraph{{\bf Base Case 1:} $\sketch = \thole{\tau}$} Let $\regex$  be a completion of $\sketch$ that satisfies $\ex^+$. Since $\sketch$ is a single hole, we apply either \textsc{Hole-Feasible} or \textsc{Hole-inFeasible} rule for doing decomposition. 
\begin{itemize}
    \item {\bf if $\forall_{e \in \ex^+}. {\tt SemanticType}(e) \subtype \tau$}, then we apply the \textsc{Hole-Feasible} rule and obtain the decomposition $[\thole{\tau} \mapsto \ex^+]$. Since we know $\regex$ satisfies $\ex^+$, we obtain a decomposition such that $\regex$ satisfies $\goal[\holesym]$.
    \item {\bf if $\exists_{e \in \ex^+}. {\tt SemanticType}(e) \not\subtype \tau$}, then using Lemma~\ref{lemma:type_semantic_correspondence}, we know that there does not exist a $\regex$ such that it accepts $\ex^+$, which contradicts the assumption. 
\end{itemize}

\paragraph{{\bf Base Case 2:} $\sketch = \regex$, where $\regex$ is a concrete regex} Since there is no hole in this sketch, and we know $\regex$ satisfies $\ex^+$, following the rule \textsc{Concrete-Feasible}, we obtain an empty decomposition so this case is vacuously true. 

\paragraph{{\bf Inductive hypothesis:}} We assume that this theorem holds from sketch $\sketch_1, \cdots, \sketch_n$. 

\paragraph{{\bf Inductive case: }} We show that this theorem holds for all sketch that is constructed using the sketch in the inductive hypothesis. Let $\regex$ be a completion of the $\sketch = f(\sketch_1, \cdots \sketch_n)$ that satisfies $\ex^+$. Also let $\regex^\bigstar =  {\tt OverApprox}(\sketch)$, we prove this part of the theorem by dividing it into two cases:

\begin{itemize}
    \item If ${\tt Match}(\regex^\bigstar, \ex^+) \not\equiv \emptyset$, let $\regex^\bigstar = f(\regex_1^\bigstar, \cdots, \regex_n^\bigstar)$, following the definition of the over-approximation and assuming we have the precise inverse semantic of construct $f$, then there must exist $\ex^+_1, \cdots, \ex^+_n$, such that $\forall_i. \ \ex^+_i \in \sem{\regex_i^\bigstar}$ and $\{f(\ex^+_1[i], \cdots, \ex^+_n[i]) \ | \ i \in \{1..|\ex^+|\} \} \equiv \ex^+$. Using such set $\ex^+_1, \cdots \ex^+_n$, we obtain a set of decomposition $\Psi_i$ for each $\sketch_i$ using the examples $\ex^+_i$. \\
    - Assuming that for each sketch $\sketch_i$, we have a completion $\regex_i$ such that $\regex_i$ matches $\ex^+_i$. Then, following the inductive hypothesis, there exists a decomposition $\goal_i$ such that every $r_{ij}$ in $r_i$ where $j$ represents the regex for the $j$th hole, $r_{ij}$ satisfies $\goal_i[h_j]$. Once we obtain such a decomposition, we compose the decomposition of $\sketch$, $\goal$, by merging the decomposition $\goal_i$ for each $\sketch_i$ using the ${\tt Merge}$. Since $f$ is not a hole, all holes in $\sketch$ are included in $\goal$. Therefore, following the inductive hypothesis, we prove that the decomposition $\goal$ we obtained has a completion, $\regex$, where the part of the $\regex$ that corresponds to $h_i$ satisfies $\goal[h_i]$. \\
    - Assuming that there exists a $\sketch_i$ with no completion $\regex_i$ such that $\regex_i$ matches any one of the possible $\ex_i^+$. Following this assumption, as well as the assumption that $\ex^+_i$ is derived from the precise inverse semantics of $f$, we know that we can never compose a program $\regex$ using the completion of $\sketch_1, \ldots, \sketch_n$ such that $\regex$ can match the full positive examples $\ex^+$. We can therefore conclude $\sketch$ as well does not have a completion $\regex$ such that $\regex$ matches $\ex_i^+$, which contradicts the premise of the theorem. 
    \item If ${\tt Match}(\regex^\bigstar, \ex^+) \equiv \emptyset$, and since $\regex^\bigstar$ is an over-approximation, there does not exist a concrete regex $\regex'$ in the language such that $\regex'$ can match all the positive examples in $\ex^+$, which as well contradicts the premise of the theorem.
\end{itemize}

\end{proof}

\begin{lemma}\label{lemma:pruning_correctness}
    Let $\type_h, \ex^+, \ex^{\bigstar}$  be the inputs to \textsc{GetNextCompletion}, and let $\prog$ be a partial program such that there exists some complete program $\prog'$ with the most precise type $\tau_{\prog}$ that can be derived from $\prog$ such that $\tau \subtype \type_h$. If $\prog'$ can accept all the positive examples $\ex^+$, then \textsc{GetNextCompletion} will add $\prog$ to the worklist $\worklist$. 
\end{lemma}

\begin{proof}
    We prove this lemma by inducting on the number of terminals $m$ in the AST of program $\prog$. 
    \paragraph{{\bf Base Case:} m = 0} The only such program $\prog_0$ with 0 terminals is a partial program with one hole that is annotated with the goal output type $\type_h$. This program is added to $\worklist$ on line 3 of Figure~\ref{fig:hole_synthesis}. 
    \paragraph{{\bf Inductive Hypothesis:}} Assume this lemma holds for all programs whose ASTs have less than $m$ terminals, where $m \geq 0$.
    \paragraph{{\bf Inductive Case:}} Suppose $\prog_{m + 1}$ has $m + 1$ terminals. Then there is some program $\prog_{m'}$ with $m' \leq m$ terminals and some production $\alpha$ such that expanding $\prog_{m'}$ with $\alpha$ produces $\prog_{m + 1}$. 
    
    Since $\prog'$ can be derived from $\prog_{m + 1}$ and therefore can also be derived from $\prog_{m'}$.By inductive hypothesis, $\prog_{m'}$ is added to $\worklist$. Then at some point $\prog_{m'}$ will be dequeued from $\worklist$ on line 5 of Figure~\ref{fig:hole_synthesis}. The ${\tt Expand}$ procedure on line 8 will identify $\alpha$ as a possible production, and will expand $\prog_{m'}$ to $\prog_{m+1}$. Since $\prog' \subtype \type_h$, then assuming the soundness of the type propagation rules, we know that for all nodes in $\prog'$, ${\tt TypeOf}(\prog'(n)) \subtype {\tt GoalType}(n)$. Since there exists a completion from $\prog_{m+1}$ to $\prog'$, we know $\forall n \in {\tt Nodes}(\prog_{m+1}) . \ {\tt IsComplete}(\prog_{m+1}(n)) \wedge \vdash {\tt TypeOf}(\prog_{m+1}(n)) \subtype {\tt GoalType}(n)$, and therefore pass the check on line 12.

    In addition, we also check if the over-approximation of $\prog_{m+1}$ can match all $\ex^+$. Since $\prog_{m+1}$ will be instantiated to a program that accepts $\ex^+$, and over-approximation of a partial program will accept $\ex^+$ if any instantiation of the partial program can accept $\ex^+$, $\prog_{m+1}$ passes the check on line 15 and eventually be added to the worklist $\worklist$ on line 16. 
\end{proof}

\begin{theorem}
Let $R$ be the set of solutions returned by {\sc GetNextCompletion}$(\tau_h, \mathcal{E}^+, \mathcal{E}_\star)$. We have:
\begin{itemize}[leftmargin=*]
    \item {\bf Soundness:} Every $r \in R$ is a solution to the hole synthesis problem, meaning (1) $r$ has type $\tau_h$ and (2) satisfies examples $\mathcal{E}^+$
    \item {\bf Completeness:} If $r \not \in R$, then $r$ is either not a solution or is observationally equivalent to some $r' \in R$ for strings in $\mathcal{E}_\star$. 
\end{itemize}
\end{theorem}

\begin{proof}
    We first prove the soundness of the algorithm. To return a concrete program $\prog$, we check if $\prog$ has a type that subtypes 
$\tau_h$ on line 7 of Figure~\ref{fig:hole_synthesis}, which proves point (1); Furthermore, we also check if $\prog$ matches all the positive examples  on line 7 of Figure~\ref{fig:hole_synthesis}, which proves point (2). 

We now prove the completeness of the algorithm. By Lemma~\ref{lemma:pruning_correctness}, any partial program $\prog$ that might expand to a solution (i.e. (1) $r$ has type $\tau_h$ and (2) satisfies examples $\mathcal{E}^+$) is added to the worklist $\worklist$. Note that the termination criteria for \textsc{GetNextCompletion} is when the worklist $\worklist$ is exhausted. Thus, $\prog$ will be dequeued at line 5 at some time during the synthesis procedure. In line 7, we first check to ensure that $\prog$ is indeed a correct solution, and then in line 8, we also check if $\prog$ is observationally equivalent to some $\regex \in R$ with respect to the set $\ex_\bigstar$. Since line 9 is the only place we return $\prog$, we obtain all the solutions that are (1) correct and (2) not observationally equivalent to other programs in the set when \textsc{GetNextCompletion} terminates. 
\end{proof}

\section{Additional Semantics of the DSL}

We provide the semantics of the string transformation part of the DSL in Figure~\ref{fig:dsl_semantics_transformation}.

\begin{figure}[ht]
    \[
    \begin{array}{r l}
    \sem{{\tt id}}s = & s \\
    \sem{{\tt toLower}}s = & {\tt ToLowerCase}(s) \\
    \sem{{\tt toUpper}}s = & {\tt ToCapitalCase}(s) \\ 
    \sem{{\tt substring}[k_1, k_2]}s = & s[k_1:k_2] \\
    \sem{{\tt abbreviate}[c]}s = & {\tt reduce}(\lambda x, y. \ x \cdot c \cdot y, {\tt split}(s)) \cdot c
    \end{array}
    \]
    \caption{Semantics of string transformation part of the DSL.}   
    \label{fig:dsl_semantics_transformation}
\end{figure}

\section{Typing Rules}

We provide the rest of the typing rules in Figure~\ref{fig:typing_rules_more}.

\begin{figure}[ht]
\vspace{-0.5cm}
\small
\begin{mathpar}
    \inferrule*[Left=Star-1]{\vdash r: {\sf CharSeq}}{\vdash r* : {\sf CharSeq}}\and 
    \inferrule*[Left=Star-2]{\vdash r: \type \ \ \ \type \not\subtype {\sf CharSeq}}{\vdash r* : {\sf Any}} \\
    \inferrule*[Left=Plus-1]{\vdash r: {\sf CharSeq}}{\vdash r+ : {\sf CharSeq}}\and 
    \inferrule*[Left=Plus-2]{\vdash r: \type \ \ \ \type \not\subtype {\sf CharSeq}}{\vdash r+ : {\sf Any}} \\
    \inferrule*[Left=RepeatRange-1]{\vdash r: {\sf CharSeq}}{\vdash r\{k_1, k_2\}: {\sf CharSeq}}\and
    \inferrule*[Left=RepeatRange-2]{\vdash r: \type \ \ \ \type \not\subtype {\sf CharSeq}}{\vdash r\{k_1, k_2\}: {\sf Any}} \\
    \inferrule*[Left=Repeat-1]{\vdash r: {\sf CharSeq}}{\vdash r\{k_1\} : {\sf CharSeq}}\and
    \inferrule*[Left=Repeat-2]{\vdash r: \type \ \ \ \type \not\subtype {\sf CharSeq}}{\vdash r\{k_1\} : {\sf Any}}

\end{mathpar}
\caption{Additional typing rules.}    
\label{fig:typing_rules_more}
\vspace{-0.5cm}
\end{figure}

\section{User Study Procedure}

In this section, we describe our user-study protocol in more detail. 

\paragraph{User study sessions} Our user study was completed in $13$ sessions, one for each participant. The participants used the same laptop, with the tool installed, across all sessions. 

\paragraph{Participant introduction} We started each user study session by first giving a general description of the task and the goal of the task. In particular, we asked them to complete 4 tasks using either standard regexes or semantic regexes. The specification for each task is a set of positive examples and a set of negative examples and the goal is to write a generalizable program that differentiate positive examples from negative examples. In order to minimize the effect of knowledge transfer, we randomly determined whether a participant was first given a task using standard regex or using semantic regex.

\begin{table}[t]
    \centering
    \scriptsize
    \begin{tabular}{cp{40mm}p{40mm}}
    \toprule
         Task Description & Sample Positive Examples &  Sample Negative Examples \\
    \midrule
         Business name with a store id & 24 HOUR FITNESS, INC. \#547 \newline 7 Eleven \#2366-24139C \newline BURGER KING 4525 & 1601 Bar \& Kitchen \newline 24 Hour Fitness, Inc \newline AT\&T - DOGGIE DINER room 3228 \\
    \midrule
        TVs of size >= 42’ and resolution >= 2160P &  43" Class LED 2160p \newline 49" Class OLED 2160p \newline 60" Class LED 4320p & 48" Class LED 1080p \newline 40" Class LED 2160p \newline 50" Class LED 1080p \\
    \midrule
        Product that contains measurement information & SWIRLY MARBLES BAG 125 g \newline VINTAGE FOLDING RULER 50 CM \newline BLACK CHRISTMAS TREE 30 CM & VINTAGE TREE 2 CM \newline CAT 12 LBS \newline DISTANCE FROM HOME 200 M \\
    \midrule
        European artists born before 1900 & French, 1610-1686|Italian, 1635-1688 \newline Italian, 1555-1630|Dutch, 1586-1652 \newline French, 1751-1832 & Swiss, 1901-1966|Swiss, 1902-1985 \newline Mexican, 1900-1940 \newline American, 1786-1877 \\
    \bottomrule
    \end{tabular}
    \caption{Descriptions of task used in the user study}
    \label{tab:user_study_task}
\end{table}

\paragraph{Task selection} We randomly selected 4 tasks from all the tasks we have in the benchmark. To ensure the tasks can be finished within 5 minutes, we slightly simplified the task. The description of the tasks and the provided sample positive and negative examples are presented in Table~\ref{tab:user_study_task}. 

\paragraph{Training} We start the training procedure walking the user through a regex ``cheatsheet'' that contains the syntax and semantics for both standard regexes and semantic regexes, as well as some sample programs in each representation. After users are comfortable with both types of regexes, we demonstrate the workflow of a task by walking the user through a following training task:

\footnotesize
\begin{tabular}{cp{35mm}p{35mm}}
    \toprule
         Task Description & Sample Positive Examples &  Sample Negative Examples \\
    \midrule
         Place associated with a year earlier than 1960 & France 1958 \newline New York 1624 \newline Tokyo 1868 & Google 2008 \newline Meta 2004 \newline Palm Desert 1973 \\
    \bottomrule
\end{tabular}
\normalsize

Initially, we present the task to the user in the context of composing a standard regex. Upon displaying the prompt on the interface, we guide the user in crafting a regex for the specific task by elucidating the distinctions between positive and negative examples. Furthermore, we demonstrate the output generated when executing their program on the task's strings within the interface. We emphasize the presence of the unseen test set and remind users that their objective also includes optimizing performance on this test set. We go through this training task for writing a semantic regex in a similar way. We finish the training procedure by asking if the user has additional questions about the procedure.

\paragraph{User study workflow} Upon completing the training procedure, the user commences work on the tasks. We notify the user at the beginning of each task whether they should use a standard regex or a semantic regex to solve it. The command-line interface records the user's inputs throughout the study, selecting their program with the highest performance on the unseen test set as their final submission. We terminate the task early if the user either resolves the task (on both the sample set and testing set) or determines they cannot improve upon their current program. Otherwise, we conclude the task after the 5-minute time limit and proceed to the next task. The whole study took around 35 minutes per participant.

\section{Benchmark Descriptions}

The description of all the tasks used in the evaluation is shown in Table~\ref{tab:benchmarks}.

\begin{table}[t!]
    \centering
    \scriptsize
    \begin{tabular}{cc}
    \toprule
    {\bf Domain} & {\bf Task Description} \\
    \midrule
    \multirow{7}{*}[2pt]{Business} & {Businesses located in California}  \\
    & {Business located in city with names started with a capital letter}  \\
    & {Owner names that is associated with a company}  \\
    & {Store names with id}  \\
    & {Market names that contains street name}  \\
    & {Restaurant names that are named after a person}  \\
    & {Restaurant that are created before 2000 or after 2010}  \\
    \midrule
    \multirow{6}{*}{Sales} & {Products with Intel CPU that has more than 8GB memory}  \\
    & {AT\&T or Verizon phone that has more than 32 GB memory}  \\
    & {DSLR cameras with lens of focal length between 18 and 200mm}  \\
    & {TVs of diagonal size more than 42' and resolution greater than 2160P}  \\
    & {TVs of diagonal size less than 50' or resolution less than 1080P}  \\
    & {Purchase dates that occur in the evening of May}  \\
    \midrule
    \multirow{6}{*}{Retail} & {Website titles with three categories separated by `|'}  \\
    & {Website titles that start with product names and followed by a url}  \\
    & {Product names that contain measurement information}  \\
    & {Product that contains at least 6 sets of items}  \\
    & {Products names that contain color information}  \\
    & {Jewelry names that contain color information}  \\
    \midrule
    \multirow{4}{*}{Marketing} & {Software engineers job with salary > 100k}  \\
    & {Software engineers job that have specified working locations}  \\
    & {Business names with at least 3 words}  \\
    & {Business names in the format of a certain agency of a place}  \\
    \midrule
    \multirow{3}{*}{Account} & {Emails that contains numbers in the username part}  \\
    & {Emails in a country domain and with username ends with number}  \\
    & {Software versions with at least 10 minor updates and more than one patches}  \\
    \midrule
    \multirow{2}{*}{Stock} & {Company names with 3-letter abbreviation}  \\
    & {Company names with ticker symbols containing special characters}  \\
    \midrule
    \multirow{5}{*}{Science} & {Locations description of format State; County; More details}  \\
    & {Location description containing distance information}  \\
    & {Locations that are either in forests or in parks}  \\
    & {Locations that are around lakes}  \\
    & {Locations that are less than 11 miles from a road}  \\
    \midrule
    \multirow{4}{*}{Server} & {Apache log with file id >=151000 or in the format of a zip file with id <= 50}  \\
    & {Files with at least 2 directory deep}  \\
    & {PHP files}  \\
    & {Photo files with numbers in their names}  \\
    \midrule
    \multirow{6}{*}{Museum} & {Gift from two people after 2000}  \\
    & {Gift from some institution before 2000}  \\
    & {Purchase made by using 3 different funds}  \\
    & {Artwork with two artists, both born in 14th century}  \\
    & {Artwork with artists all from European countries}  \\
    & {Artists biography that at least contain country and year information, may contain born city}  \\
    \midrule
    \multirow{5}{*}{Exhibition} & {Dimension of item with length, height and width all between 10 and 50 inches}  \\
    & {Dimension of item with height greater than 10 inch and is described in a specific format}  \\
    & {Dimension of item with both height and diameter less than 10 inches}  \\
    & {Item that is associated with at least three categories}  \\
    & {Item that is associated with a set of categories where each category is a single word}  \\
    \bottomrule
    \end{tabular}
    \caption{Description of the tasks used in the evaluation.}\label{tab:benchmarks}
\end{table}

\section{Prompting Details}

\noindent The prompt for the sketch generation is shown in Figure~\ref{fig:full_prompt_sketch_gen}.

\noindent The prompt for the semantic matching is shown in Figure~\ref{fig:prompt_entity_recognition}.

\noindent  The prompt for running the baseline \textsc{ChatGPT-Exec} is shown in Figure~\ref{fig:prompt_chatgpt_exec}.

\noindent The prompt for running the baseline \textsc{ChatGPT-Regex-Synth} is shown in Figure~\ref{fig:prompt_chatgptregexsynth_1} and Figure~\ref{fig:prompt_chatgptsynth_2}.

\noindent  The prompt for running the baseline \textsc{ChatGPT-Synth} is shown in Figure~\ref{fig:prompt_chatgptsynth_1}, Figure~\ref{fig:prompt_chatgptsynth_2} and Figure~\ref{fig:prompt_chatgptsynth_3}.

\begin{figure}
    \centering
    \scriptsize
    \begin{tabularx}{\linewidth}{|X|}
    \toprule
    \multicolumn{1}{|c|}{Prompt for Sketch Generation} \\
    \midrule
    Summarize the structure of the following positive examples in the form of a regular expression sketch. Use \{??: <semantic type>\} to represent the unknown part of the sketch. \\
    Positive examples: \\
    - (David J. Alexander), Marc Henri Sempere and Jocelyn Bulow \\
    - (Connie Wong), Sai Wong\\
    - (Amin Abughosh) and Joseph Abughosh and Abeer Elafifi \\
    Sketch: \\
    - \textbackslash(\{??: Person\}\textbackslash) ((\&|and|,) \{??: Person\})+ \\
    \\
    Summarize the structure of the following positive examples in the form of a regular expression sketch. Use \{??: <semantic type>\} to represent the unknown part of the sketch. \\
    Positive examples: \\
    - Arugello Market Corp. \\
    - HollyFrontier Corporation \\
    - Iron Pan, Inc. \\
    Sketch: \\
    - \{??: Company Name\} (, Inc|\{??: Corporation\})?(\textbackslash.)? \\
    \\
    Summarize the structure of the following positive examples in the form of a regular expression sketch. Use \{??: <semantic type>\} to represent the unknown part of the sketch. \\
    Positive examples: \\
    - Bistro Burger Market Street \\
    - Coffeeshop - 3139 Mission \\
    - Crab Station at Fisherman's Wharf \\
    Sketch:\\
    - \{??: Restaurant\} ((-|at) )?\{??: Location\} \\
    \\
    Summarize the structure of the following positive examples in the form of a regular expression sketch. Use \{??: <semantic type>\} to represent the unknown part of the sketch. \\
    Positive examples: \\
    - 15. Mugs \& Cups | Drinkware | Google Merchandise Store \\
    - 15. Bags | Google Merchandise Store \\
    - 10. Men's Outerwear | Apparel | Google Merchandise Store \\
    Sketch:\\
    - \{??: Integer\}\textbackslash. \{??: Product\} \textbackslash|(\{??: Category\} \textbackslash|)?Google Merchandise Store \\
    \\
    Summarize the structure of the following positive examples in the form of a regular expression sketch. Use \{??: <semantic type>\} to represent the unknown part of the sketch. \\
    Positive examples: \\
    - Gift of Robert McBratney and Company|1929 \\
    - Gift of Minic Custom Woodwork, Inc. New York|1983 \\
    - Purchase, Edward C. Moore Jr. Gift|1923 \\
    Sketch:\\
    - (Purchase, )?\{??: Gift\}\textbackslash|\{??: Date\} \\
    \\
    Summarize the structure of the following positive examples in the form of a regular expression sketch. Use \{??: <semantic type>\} to represent the unknown part of the sketch. \\
    Positive examples: \\
    - 0.5 m (50 cm)\\
    - 1.55 kg (1550 g)\\
    - .5 cm (50 mm)\\
    Sketch:\\
    - \{??: Float\} \{??: Unit\} \textbackslash(\{??: Float\} \{??: Unit\}\textbackslash)\\
    \\
    Summarize the structure of the following positive examples in the form of a regular expression sketch. Use \{??: <semantic type>\} to represent the unknown part of the sketch. \\
    Positive examples:\\
    - 0.5 m, 50 cm\\
    - 0.05 m, 5 cm\\
    - 0.05 m, 0.5 cm\\
    Sketch:\\
    - \{??: Float\} m, \{??: Float\} cm\\
    \\
    Summarize the structure of the following positive examples in the form of a regular expression sketch. Use \{??: <semantic type>\} to represent the unknown part of the sketch. \\
    Positive examples:\\
    - Director of DevOps,R\&D,54,53,53,16,63,17\\
    - Head of People Ops,Finance \& Operations,,10,10,2,4,2\\
    - Sr. Product Manager,Product,27,9,16,4,18,10\\
    Sketch:\\
    - \{??: Job\},\{??: Department\}(,\{??:Integer\})\{6\} \\
    \bottomrule
    \end{tabularx}
    \caption{Prompt for sketch generation.}
    \label{fig:full_prompt_sketch_gen}
\end{figure}

\begin{figure}
    \centering
    \scriptsize
    \begin{tabularx}{\linewidth}{|X|}
    \toprule
    \multicolumn{1}{|c|}{Prompt for Entity Recognition} \\
    \midrule
    Identify all possible substrings of the given input that has the specified semantic. Output none if you are not confident enough.\\\\
    Composite.Motors,Inc.\\
    Organization: [Composite.Motors];[Composite.Motors,Inc];[Composite.Motors,Inc.] \\\\
    Composite.Motors,Inc. \\
    Person: none\\\\
    Big Data Architect at Madison, WI\\
    Place: [Madison];[WI];[Madison, WI]\\\\
    470-43" Class (42.5" Diag.)   LED   1080p\\
    Integer: [470];[43];[1080]\\\\
    2011-03-02\\
    Date: [2011-03-02]\\\\
    1955-10-18\\
    Date: [2011]\\\\
    404-Stream 11.6" Laptop   Intel Celeron   2GB Memory\\
    Product: [Stream 11.6" Laptop];[Intel Celeron]\\\\
    Set 2 Tea Towels I Love London\\
    Item: [Tea Towels][Tea Towels I Love London]\\
    \bottomrule
    \end{tabularx}
    \caption{Prompt for entity recognition.}
    \label{fig:prompt_entity_recognition}
\end{figure}
\begin{figure}
    \centering
    \scriptsize
    \begin{tabularx}{\linewidth}{|X|}
    \toprule
    \multicolumn{1}{|c|}{Prompt for \textsc{ChatGPT-Exec}} \\
    \midrule
    Given the string below, output `Yes' if this string should be matched, `No' if this string should not be matched.\\\\
    Training\_positive\_example\_1 \\
    Matched? Yes\\
    \\$\cdots$\\\\
    Training\_negative\_example\_1 \\
    Matched? No\\
    \\$\cdots$\\\\
    Testing\_example\\
    Matched?\\
     \bottomrule
    \end{tabularx}
    \caption{Prompt for \textsc{ChatGPT-Exec}.}
    \label{fig:prompt_chatgpt_exec}
\end{figure}
\begin{figure}
    \centering
    \scriptsize
    \begin{tabularx}{\linewidth}{|X|}
    \toprule
    \multicolumn{1}{|c|}{Prompt for Running \textsc{ChatGPT-Regex-Synth} I} \\
    \midrule
    Find a program using a regular expression such that the program can match the positive examples and reject the negative examples.\\
    Positive examples: \\
    - (David J. Alexander), Marc Henri Sempere and Jocelyn Bulow \\
    - (Connie Wong), Sai Wong\\
    - (Amin Abughosh) and Joseph Abughosh and Abeer Elafifi \\
    Negative examples: \\
    - Connie Wong, Sai Wong\\
    - Amin Abughosh\\
    - Chilli House Inc.\\
    Program:\\
    - ([\textbackslash w .]+)(,)?( and)? [\textbackslash w .]+\\\\
    Find a program using a regular expression such that the program can match the positive examples and reject the negative examples.\\
    Positive examples: \\
    - Arugello Market Corp. \\
    - HollyFrontier Corporation \\
    - Iron Pan, Inc. \\
    Negative examples: \\
    - WONG JUDITH L \\
    - South Seattle\\
    - Brass Instrument Lubricants\\
    Program:\\
    - [\textbackslash w .]+,? (Corp|Inc)[.]?(oration)?\\\\
    Find a program using a regular expression such that the program can match the positive examples and reject the negative examples.\\
    Positive examples: \\
    - Bistro Burger Market Street \\
    - Coffeeshop - 3139 Mission \\
    - Crab Station at Fisherman's Wharf \\
    Negative examples:\\
    - 20th Century Cafe \\
    - ALL SEASON MARKET\\
    - AUTO CITY BRUSHLESS CAR WASH\\
    Program:\\
    - .* (at [\textbackslash w ']+|Street|\textbackslash d+ [\textbackslash w ']+)\\\\
    Find a program using a regular expression such that the program can match the positive examples and reject the negative examples.\\
    Positive examples: \\
    - 15. Mugs \& Cups | Drinkware | Google Merchandise Store \\
    - 15. Bags | Google Merchandise Store \\
    - 10. Men's Outerwear | Apparel | Google Merchandise Store \\
    Negative examples:\\
    - 2. Women's T-Shirts | Apparel | Google Merchandise Store\\
    - 22. Water Bottles \& Tumblers | Drinkware | Google Merchandise Store\\
    - Google Women's Yoga Pants
    Program:\\
    - 1\textbackslash d\textbackslash . .*[|].*[|]?.*\\\\
    Find a program using a regular expression such that the program can match the positive examples and reject the negative examples.\\
    Positive examples: \\
     - Gift of Robert McBratney and Company|1929 \\
    - Gift of Minic Custom Woodwork, Inc. New York|1983 \\
    - Purchase, Edward C. Moore Jr. Gift|1923 \\
    Negative examples:\\
    - Fletcher Fund, 1941 \\
    - Gift of Emma and Jay A. Lewis|2004\\
    - The Michael C. Rockefeller Memorial Collection, Gift of Harry M. Miller Jr., and Professor Paulo de Goes, 1965\\
    Program:\\
    - (Purchase, )?(Gift of .*|.* Gift)|1\textbackslash d\{3\}\\\\
    Find a program using a regular expression such that the program can match the positive examples and reject the negative examples.\\
    Positive examples:\\
     - 0.5 m (50 cm)\\
    - 1.55 kg (1550 g)\\
    - .5 cm (50 mm)\\
    Negative examples:\\
    - 0.6 m (60 cm) \\
    - 2.20 kg (2200 g)\\
    - .8 cm (80 mm)\\
    Program:\\
    - \textbackslash d*[.]\textbackslash d*5 (m|cm|kg) (\textbackslash d*5\textbackslash d* (cm|g|mm))\\
    \bottomrule
    \end{tabularx}
    \caption{Prompt for the \textsc{ChatGPT-Regex-Synth} I.}
    \label{fig:prompt_chatgptregexsynth_1}
     \end{figure}
     
    \begin{figure}
    \centering
    \scriptsize
    \begin{tabularx}{\linewidth}{|X|}
    \toprule
    \multicolumn{1}{|c|}{Prompt for Running \textsc{ChatGPT-Regex-Synth} II} \\
    \midrule
     Find a program using a regular expression such that the program can match the positive examples and reject the negative examples.\\
     Positive examples:\\
     - 0.5 m, 50 cm\\
    - 0.05 m, 5 cm\\
    - 0.05 m, 0.5 cm\\
    Negative examples:\\
    - 0.6 m, 60 cm\\
    - 0.05 m (5 cm)\\
    - .8 cm, 80 mm\\
    Program:
    - \textbackslash d*.\textbackslash d*5 m, \textbackslash d*.?5\textbackslash d* cm \\\\
     Find a program using a regular expression such that the program can match the positive examples and reject the negative examples.\\
     Positive examples:\\
     - Director of DevOps,R\&D,54,53,53,16,63,17\\
    - Head of People Ops,Finance \& Operations,,10,10,2,4,2\\
    - Sr. Product Manager,Product,27,9,16,4,18,10\\
    Negative examples:\\
    - Director of DevOps,R\&D,54,a,53,16,63,17\\\
    - Product, Sr. Product Manager,27,9,16,4,18,10\\
    - Sr. Product Manager,Product,27,9,16,4,18,10,12,13\\
    Program:\\
    - [\textbackslash w \&.]+,(R\&D|Finance \& Operations|Product)(,\textbackslash d*)\{6\}\\
    \bottomrule
    \end{tabularx}
    \caption{Prompt for the \textsc{ChatGPT-Regex-Synth} II.}
    \label{fig:prompt_chatgptregexsynth_2}
\end{figure}
\begin{figure}
    \centering
    \scriptsize
    \begin{tabularx}{\linewidth}{|X|}
    \toprule
    \multicolumn{1}{|c|}{Prompt for Running \textsc{ChatGPT-Synth} I} \\
    \midrule
    ``` \\
    Semantic Regex Syntax:\\
    r ::= constant | cc\\
    \ \ \     | \{<type> -> f\} | \{<$\texttt{type}_b$> -> p\} | \{<$\texttt{type}_b$>\}\\
    \ \ \     | r? | r* | r+ | r\{n\} | r\{n1,n2\}\\
    \ \ \     | rr | r|r | r \& r\\
    f ::= x | toUpper | toLower\\
    \ \ \     | substring[number1, number2] | abbreviate[string]\\ 
    p ::= True | ~p | p|p | p\&p | NumMatch(number1, sym, number2, sym)\\ 
    \ \ \     | isYear(year1, year2) | isMonth(month1, month2) | isDate(date1, date2)\\
    \ \ \    | btwHour(n1, n2) | btwMin(n1, n2) | btwSec(n1, n2) | isMorning | isAfternoon | isEvening\\
    \ \ \     | inRegion(continent) | inCountry(country) | inState(state)\\
    cc ::= ANY | LET | NUM | CAP\\
    $\texttt{type}_b$ ::= Person | Organization | Product | Event | Work of Art\\
        \ \ \     | Number | Integer | Float\\
        \ \ \     | Date | Year | Month | Day\\
        \ \ \     | Time | Hour | Minute | Second\\
        \ \ \     | Place | Location | Nationality | Country | City\\
    ```\\
    Find a program using a semantic regular expression such that the program can match the positive examples and reject the negative examples.\\
    Positive examples: \\
    - (David J. Alexander), Marc Henri Sempere and Jocelyn Bulow \\
    - (Connie Wong), Sai Wong\\
    - (Amin Abughosh) and Joseph Abughosh and Abeer Elafifi \\
    Negative examples: \\
    - Connie Wong, Sai Wong\\
    - Amin Abughosh\\
    - Chilli House Inc.\\
    Program:\\
    - (\{<Person>\}) ((\&|and|,) \{<Person>\})+ \\
        \bottomrule
    \end{tabularx}
    \caption{Prompt for running \textsc{ChatGPT-Synth} I. }
    \label{fig:prompt_chatgptsynth_1}
\end{figure}

\begin{figure}
    \centering
    \scriptsize
    \begin{tabularx}{\linewidth}{|X|}
    \toprule
    \multicolumn{1}{|c|}{Prompt for running \textsc{ChatGPT-Synth} II} \\
    \midrule
    Find a program using a regular expression such that the program can match the positive examples and reject the negative examples.\\
    Positive examples: \\
    - Arugello Market Corp. \\
    - HollyFrontier Corporation \\
    - Iron Pan, Inc. \\
    Negative examples: \\
    - WONG JUDITH L \\
    - South Seattle\\
    - Brass Instrument Lubricants\\
    Program:\\
    - \{<Company>\} (, Inc|\{<Corporation>\})?(.)?\\\\
    Find a program using a regular expression such that the program can match the positive examples and reject the negative examples.\\
    Positive examples: \\
    - Bistro Burger Market Street \\
    - Coffeeshop - 3139 Mission \\
    - Crab Station at Fisherman's Wharf \\
    Negative examples:\\
    - 20th Century Cafe \\
    - ALL SEASON MARKET\\
    - AUTO CITY BRUSHLESS CAR WASH\\
    Program:\\
    - \{<Restaurant>\} ((-|at) )?\{<Location>\}\\\\
    Find a program using a regular expression such that the program can match the positive examples and reject the negative examples.\\
    Positive examples: \\
    - 15. Mugs \& Cups | Drinkware | Google Merchandise Store \\
    - 15. Bags | Google Merchandise Store \\
    - 10. Men's Outerwear | Apparel | Google Merchandise Store \\
    Negative examples:\\
    - 2. Women's T-Shirts | Apparel | Google Merchandise Store\\
    - 22. Water Bottles \& Tumblers | Drinkware | Google Merchandise Store\\
    - Google Women's Yoga Pants
    Program:\\
    - \{<Integer> -> NumMatch(10, <=, 20, <=)\}. \{<Product>\} [|] (\{<Category>\} [|])?Google Merchandise Store\\\\
    Find a program using a regular expression such that the program can match the positive examples and reject the negative examples.\\
    Positive examples: \\
     - Gift of Robert McBratney and Company|1929 \\
    - Gift of Minic Custom Woodwork, Inc. New York|1983 \\
    - Purchase, Edward C. Moore Jr. Gift|1923 \\
    Negative examples:\\
    - Fletcher Fund, 1941 \\
    - Gift of Emma and Jay A. Lewis|2004\\
    - The Michael C. Rockefeller Memorial Collection, Gift of Harry M. Miller Jr., and Professor Paulo de Goes, 1965\\
    Program:\\
    - (Purchase, )?\{<Gift>\}, \{<Location>\}?[|]\{<Date> -> InYear(0,2000)\}\\\\
    Find a program using a regular expression such that the program can match the positive examples and reject the negative examples.\\
    Positive examples:\\
     - 0.5 m (50 cm)\\
    - 1.55 kg (1550 g)\\
    - .5 cm (50 mm)\\
    Negative examples:\\
    - 0.6 m (60 cm) \\
    - 2.20 kg (2200 g)\\
    - .8 cm (80 mm)\\
    Program:\\
    - \{<Float>\}\&.*[5] \{<Unit>\} (\{<Float>\} \{<Unit>\})\\
    \bottomrule
    \end{tabularx}
    \caption{Propt for running \textsc{ChatGPT-Synth II.}}
    \label{fig:prompt_chatgptsynth_2}
\end{figure}

  \begin{figure}
    \centering
    \scriptsize
    \begin{tabularx}{\linewidth}{|X|}
    \toprule
    \multicolumn{1}{|c|}{Prompt for Running \textsc{ChatGPT-Synth} III} \\
    \midrule
     Find a program using a regular expression such that the program can match the positive examples and reject the negative examples.\\
     Positive examples:\\
     - 0.5 m, 50 cm\\
    - 0.05 m, 5 cm\\
    - 0.05 m, 0.5 cm\\
    Negative examples:\\
    - 0.6 m, 60 cm\\
    - 0.05 m (5 cm)\\
    - .8 cm, 80 mm\\
    Program:
    - \{<Float>\}\&.*[5] m, \{<Float>\} cm \\\\
     Find a program using a regular expression such that the program can match the positive examples and reject the negative examples.\\
     Positive examples:\\
     - Director of DevOps,R\&D,54,53,53,16,63,17\\
    - Head of People Ops,Finance \& Operations,,10,10,2,4,2\\
    - Sr. Product Manager,Product,27,9,16,4,18,10\\
    Negative examples:\\
    - Director of DevOps,R\&D,54,a,53,16,63,17\\\
    - Product, Sr. Product Manager,27,9,16,4,18,10\\
    - Sr. Product Manager,Product,27,9,16,4,18,10,12,13\\
    Program:\\
    - \{<Job>\},\{<Department>\}(,\{<Integer>\}?)\{6\}\\
    \bottomrule
    \end{tabularx}
    \caption{Prompt for the \textsc{ChatGPT-Synth} III.}
    \label{fig:prompt_chatgptsynth_3}
\end{figure}

\end{document}